\documentclass[journal]{IEEEtran}
\usepackage{bbm}
\usepackage{amsmath,amssymb}
\usepackage[dvips]{graphicx}
\usepackage{amsfonts}
\usepackage[mathscr]{eucal}
\usepackage{latexsym}
\usepackage{amsthm}
\usepackage{exscale}
\usepackage[mathscr]{eucal}
\usepackage{bm}
\usepackage[dvipsnames]{color}
\usepackage{cases}
\usepackage{color}
\usepackage{epsfig}
\usepackage[center,small]{caption}
\usepackage{algorithm}
\usepackage{algorithmic}
\usepackage[verbose,nospace,sort]{cite}
\usepackage{tabularx}
\usepackage{multirow}

\graphicspath{{./figs/}}

\scrollmode

\IEEEoverridecommandlockouts

\newtheorem{proposition}{Proposition}
\newtheorem{definition}{Definition}
\newtheorem{theorem}{Theorem}

\newtheorem{lemma}{Lemma}
\newcommand{\xlog}{\mathrm{log}}

\newcommand{\bH}{\mathbf{H}}

\newcommand{\bx}{\mathbf{x}}
\newcommand{\by}{\mathbf{y}}
\newcommand{\bd}{\mathbf{d}}
\newcommand{\bs}{\mathbf{s}}
\newcommand{\bn}{\mathbf{n}}
\newcommand{\bC}{\mathbf{C}}
\newcommand{\xC}{\widetilde{C}}
\newcommand{\xbC}{\widetilde{\mathbf{C}}}
\newcommand{\lbC}{\underline{\mathbf{C}}}

\newcommand{\xE}{\mathbb{E}}
\newcommand{\xTr}{\mathrm{Tr}}

\begin{document}
\title{Transmit Optimization with Improper Gaussian Signaling for Interference Channels}
\author{Yong~Zeng,~\IEEEmembership{Student~Member,~IEEE,}
        Cenk~M.~Yetis,~\IEEEmembership{Member,~IEEE,}
        Erry~Gunawan,~\IEEEmembership{Member,~IEEE,}
        Yong~Liang~Guan,~\IEEEmembership{Member,~IEEE,}
        and~Rui~Zhang,~\IEEEmembership{Member,~IEEE}
\thanks{Y.~Zeng, E.~Gunawan, and Y.~L.~Guan are with the School of Electrical and Electronic Engineering (EEE), Nanyang Technological University (NTU), Singapore (e-mail: ze0003ng@e.ntu.edu.sg, \{egunawan, eylguan\}@ntu.edu.sg).}
 \thanks{C.~M.~Yetis was with the School of EEE, NTU. He is now with Mevlana University, Turkey (e-mail: cenkmyetis@gmail.com).}
   \thanks{R.~Zhang is with the Electrical and Computer Engineering Department, National University of Singapore (e-mail: elezhang@nus.edu.sg).}
 }


\maketitle
\begin{abstract}
This paper studies the achievable rates of Gaussian interference channels with additive white Gaussian noise (AWGN), when \emph{improper} or circularly \emph{asymmetric} complex Gaussian signaling is applied. For the Gaussian multiple-input multiple-output interference channel (MIMO-IC) with the interference treated as Gaussian noise, we show that the user's achievable rate can be expressed as a summation of the rate achievable  by the conventional \emph{proper} or circularly \emph{symmetric} complex Gaussian signaling in terms of the users' transmit \emph{covariance} matrices, and an additional term, which is a function of both the users' transmit covariance and \emph{pseudo-covariance} matrices. The additional degrees of freedom in the pseudo-covariance matrix, which is conventionally set to be zero for the case of proper Gaussian signaling, provide an opportunity to further improve the achievable rates of Gaussian MIMO-ICs  by employing improper Gaussian signaling. To this end, this paper proposes \emph{widely linear precoding}, which efficiently  maps proper information-bearing signals to improper transmitted signals at each transmitter for any given pair of transmit covariance and pseudo-covariance matrices. In particular, for the case of two-user Gaussian single-input single-output interference channel (SISO-IC), we propose a joint covariance and pseudo-covariance optimization algorithm with improper Gaussian signaling to achieve the Pareto-optimal rates. By utilizing the separable structure of the achievable rate expression, an alternative algorithm with separate covariance and pseudo-covariance optimization is also proposed,  which guarantees the rate improvement over  conventional proper Gaussian signaling. 
\end{abstract}
\section{Introduction}
The interference channel (IC) models multi-user communication systems  where each transmitter is intended to send  independent information to its corresponding receiver while causing interference to all other receivers.  Although information-theoretic study of the IC has a long history, characterization of its capacity region still remains an open problem in general, except for some special cases such as that with the presence of  ``strong'' interference \cite{272,164}. For the single-input single-output IC, termed SISO-IC, the  best achievable rate region to date is obtained by the celebrated Han-Kobayashi scheme \cite{164}.  Recently, it has been shown in \cite{163} that a particular form of this scheme achieves within one bit to the capacity region of the two-user Gaussian SISO-IC with additive white Gaussian noise (AWGN). Since such  capacity-approaching techniques require multi-user encoding/decoding, which are difficult to implement in practical systems, a more pragmatic approach is to employ single-user encoding and decoding by treating the interference as Gaussian noise at all receivers. In fact, this simplified approach  has been shown to be sum-capacity optimal for Gaussian ICs when the interference level is below a certain threshold \cite{278}.

 Under the assumption of single-user detection (SUD) with the interference treated as Gaussian noise, the transmit optimization  problem for Gaussian ICs reduces to resource allocation among the transmitters for interference mitigation, which has received significant attention in the last few decades. Early works on resource allocation for Gaussian ICs mostly focused on power control since SISO-IC with single-antenna terminals was considered (see e.g. \cite{344,253,301} and references therein). When transmitters/receivers are equipped with multiple antennas, the network performance can be further improved via transmit/receive beamforming. One useful technique applied for optimizing transmit beamformers is to transform the design problem into an equivalent receiver beamformer optimization problem via the so-called uplink-downlink or network duality principle \cite{302,305,306,307,308}. The transmit beamforming optimization problems for power minimization and signal-to-interference-plus-noise ratio (SINR) balancing  can also be directly solved  by convex optimization techniques, such as the second-order cone programming (SOCP) \cite{214} and semidefinite programming (SDP) \cite{309}.

  For Gaussian ICs with the interference treated as Gaussian noise, the sum-rate maximization problem is in general difficult to be solved globally optimally due to its non-convexity. In \cite{325}, it was shown that finding  globally optimal beamformers for the weighted sum-rate maximization (WSRMax) in the Gaussian multiple-input single-output IC (MISO-IC) is an NP-hard problem.  Algorithms based on the principle of \emph{interference pricing} have been proposed for achieving local optimums \cite{313}, while in \cite{329}, a distributed algorithm was proposed for MISO-IC by using the virtual SINR framework. Gradient descent algorithms have also been proposed for  Gaussian multiple-input multiple-output IC (MIMO-IC) over transmit covariance matrices \cite{233} or precoding matrices \cite{158}. More recently, for Gaussian SISO-IC, single-input multiple-output IC (SIMO-IC) and MISO-IC, globally optimal  solutions to the WSRMax problem have been obtained  under the \emph{monotonic optimization} framework \cite{321,322,320,345,326}. However, the complexity of such globally optimal solutions increases exponentially with the number of users, and their generalization to more general MIMO-IC remains unknown.
  An alternative technique for solving WSRMax problems for MIMO-ICs is via iteratively minimizing the weighted mean-square-error (MSE), which utilizes the inherent relationship between the mutual information and MSE \cite{333,237,335}. The study of ICs with game-theoretic models has also been given in \cite{316} and references therein. Moreover, it is worth mentioning that there has been a great deal of interest in the last few years on studying Gaussian ICs from the degrees-of-freedom (DoF) perspective \cite{92}.  A key technique to achieve higher DoFs than previously believed for Gaussian ICs is interference alignment (IA) \cite{85}.
  Since DoF only provides the approximated capacity  at the asymptotically high signal-to-noise ratio (SNR), a number of IA-based precoding schemes with improved sum-rate performance at practical SNRs have been proposed in \cite{158}, \cite{177,187,252}.

  As for characterization of the achievable rate region for Gaussian ICs with interference treated as noise, various solutions have been obtained for the SISO-IC \cite{250}, SIMO-IC \cite{320}, and MISO-IC \cite{249,240,243,323}. The Pareto boundary of the achievable rate region for ICs consists of all the achievable rate-tuples, at each of which it is impossible to improve one  user's rate, without simultaneously decreasing the rate of at least one of other users. A traditional approach for characterizing Pareto boundaries of Gaussian ICs is via solving WSRMax problems. However, as pointed out in \cite{240}, the WSRMax approach  cannot guarantee the finding of all Pareto-boundary points due to the non-convexity of the achievable rate region. An alternative method based on the concept of \emph{rate-profile} was thus proposed  in \cite{240}, which is able to characterize the complete Pareto boundary for ICs. Besides, the rate-profile approach generally results in optimization problems that  are easier to handle than  conventional WSRMax problems for ICs \cite{320,240}.

It is necessary to point out that in all the aforementioned works on Gaussian ICs, the transmitted signals are assumed to be \emph{proper} or circularly symmetric complex Gaussian (CSCG) distributed.  A key property of proper Gaussian  random vectors (RVs) is that their  second-order statistics are completely specified by the conventional covariance matrix under a zero-mean assumption. In contrast, for the more general \emph{improper} Gaussian RVs, an extra parameter called \emph{pseudo-covariance} matrix is required for the complete second-order characterization \cite{245,246,248}. Most of the existing works on Gaussian ICs have adopted the proper Gaussian assumption  without any justification. This may be due to the common practice of modeling the additive receiver noise as proper Gaussian, and the well-known maximum-entropy theorem \cite{245}, i.e., proper Gaussian RVs maximize the differential entropy for any given covariance matrix. As a result, proper Gaussian signaling has been shown to be capacity optimal for the Gaussian point-to-point channel, multiple-access channel (MAC) and broadcast channel (BC). However, for Gaussian ICs with  the interference treated as noise, improper Gaussian signaling provides a new  opportunity to further improve the achievable rates over the conventional proper Gaussian signaling \cite{348}. For instance, it was shown in \cite{189} that improper Gaussian signaling, together with symbol extensions and IA, is able to improve the DoF for the three-user SISO-IC with time-invariant channel coefficients at the asymptotically high SNR. In \cite{255,256}, it was shown that the achievable rate region can be enlarged with improper Gaussian signaling even for the two-user SISO-IC at finite SNR. This is particularly interesting since it is known that for the two-user SISO-IC, no DoF gain is achievable with IA \cite{348}.  More specifically, a suboptimal scheme was proposed in \cite{255} for the two-user SISO-IC, where the transmit covariance matrices for the equivalent real-valued MIMO-ICs are restricted to be rank-1. In \cite{256}, the Pareto-optimal transmit covariance matrices for the two-user SISO-IC are obtained by an exhaustive search method.

  The prior works \cite{189,255,256} on the study of improper Gaussian signaling for ICs  are all based on the equivalent double-sized real-valued MIMO-IC matrix by separating the real and imaginary parts of the complex-valued channels. Although any complex-valued system can be transformed into an equivalent real-valued system, as pointed out in \cite{262,246}, much of the elegancy of the system description is lost. Therefore, in this paper, we adopt the complex-valued channel model for studying improper Gaussian signaling in Gaussian ICs to gain new insights. 
 The main contributions of this paper are summarized as follows:
\begin{itemize}
\item Based on existing results on improper Gaussian RVs, we derive a new achievable rate expression for the general $K$-user MIMO-IC, when improper Gaussian signaling is applied. Our result shows that the user's achievable rate can be expressed as a summation of the rate achievable  by the conventional proper Gaussian signaling in terms of the users' transmit covariance matrices, and an additional term, which is a function of both the users' transmit covariance and pseudo-covariance matrices. This new result implies that the use of improper Gaussian signaling for MIMO-ICs with interference treated as noise is able to improve the achievable rate over the conventional proper Gaussian signaling with any given set of covariance matrices of transmitted signals, by further optimizing their pseudo-covariance matrices.
\item For any given pair of signal covariance and pseudo-covariance matrices at each transmitter, we consider the practical problem of generating improper transmitted signals from proper information-bearing signals. Based on existing techniques for improper RVs \cite{262}, we propose an efficient method for this implementation, named as \emph{widely linear precoding}. 
\item By adopting the rate-profile method, we formulate the optimization problem for the two-user SISO-IC to characterize the Pareto boundary of the achievable rate region with improper Gaussian signaling. By applying the celebrated semidefinite relaxation (SDR) technique \cite{349},  a joint covariance and pseudo-covariance optimization algorithm is proposed, which achieves near-optimal rate-pairs. Furthermore, by utilizing the separable structure of the achievable rate expression with improper Gaussian signaling, a separate covariance and pseudo-covariance optimization algorithm is also proposed, which guarantees the rate improvement over conventional proper Gaussian signaling with any given transmit covariance. 
\end{itemize}

The rest of this paper is organized as follows. Section~\ref{S:MIMO-IC} studies improper Gaussian signaling for the general MIMO-IC, where a new achievable rate expression is derived and widely linear precoding is proposed. Section~\ref{S:problemFormulation} focuses on the two-user SISO-IC setup, where the problem formulation for characterizing the Pareto boundary of the achievable rate region is given. In Section~\ref{S:jointOpt}, a SDR-based joint covariance and pseudo-covariance optimization algorithm for the two-user SISO-IC is proposed. In Section~\ref{S:separate}, an alternative SOCP-based algorithm by separate covariance and pseudo-covariance optimizations is presented. Section~\ref{S:simulation} provides numerical results. Finally, we conclude the paper in Section~\ref{S:conclusions}.

\emph{Notations:} In this paper, scalars are denoted by italic letters. Boldface lower- and upper-case letters denote vectors and matrices, respectively. $\mathbf{I}_M$ denotes an $M\times M$ identity matrix and the subscript $M$ is omitted if its value is clear from the context. $\mathbf{0}$ denotes an all-zero matrix. For a square matrix $\mathbf{S}$, $\mathrm{Tr}(\mathbf{S})$, $|\mathbf{S}|$, $\mathbf{S}^{-1}$ denote the trace, determinant and inverse of $\mathbf{S}$, respectively. $\mathbf{S}\succeq \mathbf{0}$ and $\mathbf{S}\succ \mathbf{0}$ mean that $\mathbf{S}$ is positive semidefinite and positive definite, respectively. $\mathbb{C}^{M\times N}$ and $\mathbb{R}^{M\times N}$  denote the space of $M\times N$ complex and real matrices, respectively. $\mathbb{S}^{K}$ and $\mathbb{H}^{K}$ denote the $K\times K$ real-valued symmetric and complex-valued Hermitian matrices, respectively. For an arbitrary matrix $\mathbf{A}$, $\mathbf{A}^*$, $\mathbf{A}^{T}$, $\mathbf{A}^{H}$ and $\mathrm{rank}\{\mathbf{A}\}$ represent the complex-conjugate, transpose, conjugate transpose and rank of $\mathbf{A}$, respectively. $\mathrm{diag}\{\bx\}$ represents a diagonal matrix with the elements in the main diagonal given by $\bx$. 
  $\mathcal{N}(\boldsymbol \mu,\mathbf{C})$ represents the real-valued Gaussian RV with mean $\boldsymbol \mu$ and covariance matrix $\mathbf{C}$. $\mathcal{CN}(\mathbf{x},\mathbf{\Sigma})$ represents the CSCG RV with mean $\mathbf{x}$ and covariance matrix $\mathbf{\Sigma}$. 
 For a complex number $x$, $|x|$ denotes its magnitude.
 The symbol $i$ represents the imaginary unit, i.e., $i ^2=-1$. $[\mathbf v]_k$ denotes the $k$th element of a vector $\mathbf v$. $\Re\{\cdot\}$ and $\Im \{\cdot\}$ represent the real and imaginary parts of complex numbers/matrices, respectively. $I(\mathbf x ; \mathbf y)$ represents the mutual information between two RVs $\mathbf x$ and $\mathbf y$.

\section{Improper Gaussian Signaling for MIMO-IC}\label{S:MIMO-IC}
Consider a $K$-user MIMO-IC, where each transmitter is intended to send independent information to its corresponding receiver, while possibly interfering with all other $K-1$ receivers. Denote the number of transmitting and receiving antennas for each user by $M$ and $N$, respectively. Assuming the narrow-band transmission, the equivalent baseband received signal at each receiver can be expressed as
\begin{align}\label{E:yk}
\by_k(n)=\bH_{kk}\bx_k(n)+\sum_{j\neq k} \bH_{kj}\bx_j(n)+\bn_k(n), \ \forall k,
\end{align}
where $n$ is the symbol index, $\bH_{kk}\in \mathbb{C}^{N\times M}$ denotes the direct channel matrix from transmitter $k$ to receiver $k$, while $\bH_{kj}, j\neq k$, denotes the interference channel matrix from transmitter $j$ to receiver $k$; we assume  quasi-static fading and thus all channels are constant over $n$'s in \eqref{E:yk} for the case of our interest; $\bn_k(n)$ represents the independent and identically distributed (i.i.d.) CSCG noise vector at receiver $k$ with $\bn_k(n)\sim \mathcal{CN}(\mathbf{0}, \sigma^2\mathbf{I})$; and $\bx_k(n)\in \mathbb{C}^{M}$ is the transmitted signal vector from transmitter $k$, which is independent of $\bx_j(n)$,  $\forall j\neq k$. In this paper, for the purpose of exposition, we assume that symbol extensions over time as in \cite{189} are not used. Hence, $\bx_k(n)$ is independent over $n$. For brevity, $n$ is omitted in the rest of this paper.  Different from the conventional setup where proper Gaussian signaling is assumed, i.e., $\bx_k\sim \mathcal{CN}(\mathbf 0, \mathbf C_{\bx_k})$, with $\mathbf C_{\bx_k}$ denoting the transmit covariance matrix, in this paper we consider the more general improper Gaussian signals, for which some preliminaries are given next.
\vspace{-1.5ex}
\subsection{Preliminary for Improper Random Vectors}
For a zero-mean RV $\mathbf{z}\in \mathbb{C}^{n}$, the covariance matrix $\bC_{\mathbf{z}}$ and pseudo-covariance matrix $\xbC_z$ are defined as \cite{245}
\begin{align}\label{E:CDef}
\bC_{\mathbf{z}}&\triangleq \xE(\mathbf{z}\mathbf{z}^H),\  \xbC_{\mathbf{z}}\triangleq \xE(\mathbf{z}\mathbf{z}^T).
\end{align}
By definition, it is easy to verify that the covariance matrix $\bC_{\mathbf{z}}$ is Hermitian and positive semidefinite, and the pseudo-covariance matrix $\xbC_{\mathbf{z}}$ is symmetric.
\begin{definition} \cite{245}:
 A complex RV $\mathbf{z}$ is called proper if its pseudo-covariance matrix $\xbC_{\mathbf{z}}$ vanishes to a zero matrix; otherwise, it is called improper.
\end{definition}
\begin{lemma} \label{L:uncorrelated} \cite{245}:
Two zero-mean complex RVs $\mathbf x$ and $\mathbf z$ are uncorrelated if and only if $\bC_{\mathbf{xz}}=\mathbf{0}$ and $\xbC_{\mathbf{xz}}=\mathbf{0}$, where $\bC_{\mathbf{xz}}\triangleq \mathbb{E}(\mathbf{x}\mathbf{z}^H)$ and $\xbC_{\mathbf{xz}}\triangleq \mathbb{E}(\mathbf{x}\mathbf{z}^T)$.
\end{lemma}

A more restrictive definition than properness is known as \emph{circularly} symmetric, which is defined as follows.
\begin{definition}\cite{246}:
A complex RV $\mathbf{z}$ is circularly symmetric if its distribution is rotationally invariant, i.e.,
 $\mathbf z$ and $\mathbf {\hat{z}}=e^{i \alpha }\mathbf z$ have the same distribution for any real value $\alpha$.
\end{definition}
For a circularly symmetric RV $\mathbf{z}$, we have
\begin{align}
\xbC_{\mathbf z}=\xbC_{\mathbf {\hat{z}}}=\mathbb{E}(\mathbf {\hat{z}}\mathbf {\hat{z}}^T)=e^{i 2\alpha}\xbC_{\mathbf z}, \ \forall \alpha, \notag
\end{align}
which implies  $\xbC_{\mathbf z}=\mathbf{0}$. Thus, circularity implies properness, but the converse is not true in general. However, if $\mathbf z$ is a zero-mean \emph{Gaussian} RV, then properness and circularity are equivalent, as given by the following lemma.
\begin{lemma}\cite{246}:
A complex zero-mean Gaussian RV $\mathbf z$ is circularly symmetric if and only if it is proper.
\end{lemma}
For example, the commonly adopted assumption that the noise vector $\mathbf n_k$ in \eqref{E:yk} is zero-mean CSCG is equivalent to that $\mathbf n_k$ is a proper Gaussian RV, whose pseudo-covariance matrix satisfies $\xbC_{\mathbf n_k}=\mathbf{0}$.
For an arbitrary complex RV $\mathbf{z}$, define $\lbC_{\mathbf{z}}$ as the covariance matrix of the augmented vector $[\begin{matrix} \mathbf{z}^T & {(\mathbf{z}^*)}^T \end{matrix}]^T$, i.e.,
\begin{align}\label{E:augC}
\lbC_{\mathbf{z}}\triangleq \xE \bigg( \bigg[\begin{matrix} \mathbf{z} \\ \mathbf{z}^* \end{matrix}\bigg] \bigg[\begin{matrix} \mathbf{z} \\ \mathbf{z}^* \end{matrix}\bigg]^H \bigg)=\bigg[\begin{matrix}\bC_{\mathbf{z}} & \xbC_{\mathbf{z}} \\ \xbC_{\mathbf{z}}^* & \bC_{\mathbf{z}}^* \end{matrix} \bigg].
\end{align}
The augmented covariance matrix $\lbC_{\mathbf{z}}$ obviously has some built-in redundancy for the second-order characterization of $\mathbf{z}$; however, it is useful as shown in the following two theorems.
\begin{theorem} \cite{246}: \label{T:validPair}
$\bC_{\mathbf{z}}$ and $\xbC_{\mathbf{z}}$ are a valid pair of covariance and pseudo-covariance matrices, i.e., there exists a RV $\mathbf{z}$ with covariance and pseudo-covariance matrices given by $\bC_{\mathbf{z}}$ and $\xbC_{\mathbf{z}}$, respectively,  if and only if the augmented covariance matrix $\lbC_{\mathbf{z}}$ is positive semidefinite, i.e., $\lbC_{\mathbf{z}}\succeq \mathbf{0}$.
\end{theorem}
Note that the conditions of the covariance matrix $\bC_{\mathbf{z}}$  being Hermitian and positive semidefinite, and the pseudo-covariance matrix $\xbC_{\mathbf{z}}$ being symmetric are already implied by $\lbC_{\mathbf{z}}\succeq \mathbf{0}$.
 Furthermore, for the improper complex \emph{Gaussian} RVs, the differential entropy is in general a function of both the covariance and pseudo-covariance matrices, which can be expressed in terms of $\lbC_{\mathbf{z}}$ as shown by the following theorem.
\begin{theorem} \label{T:entropyGaussian} \cite{246}:
 The  entropy of a complex Gaussian RV $\mathbf{z}\in \mathbb{C}^{n}$ with augmented covariance matrix $\lbC_{\mathbf{z}}$ is
\begin{align}
h(\mathbf{z})=\frac{1}{2}\xlog \big ((\pi e)^{2n} |\lbC_{\mathbf{z}} | \big). \label{E:hz}
\end{align}
\end{theorem}
Theorem~\ref{T:entropyGaussian} generalizes the entropy result for proper Gaussian RVs. If $\xbC_{\mathbf z}=\mathbf{0}$, \eqref{E:hz} reduces to the well-known expression for proper Gaussian RVs $h(\mathbf{z})=\xlog \big ((\pi e)^{n} |\bC_{\mathbf{z}} | \big)$ \cite{245}.
\vspace{-1.5ex}
 \subsection{Achievable Rate with Improper Gaussian Signaling}\label{S:rate}
In this subsection, we derive the achievable rate by improper Gaussian signaling for the $K$-user MIMO-IC defined in \eqref{E:yk}. Denote the covariance and pseudo-covariance matrices of the zero-mean transmitted Gaussian RV $\bx_k$ by $\bC_{\bx_k}$ and $\xbC_{\bx_k}$, respectively, i.e.,
\[
\bC_{\bx_k}=\xE(\bx_k\bx_k^H), \ \xbC_{\bx_k}=\xE(\bx_k\bx_k^T),\  k=1,\cdots, K.
\]
Since $\mathbf x_k$ and $\mathbf x_j$ are independent for $j\neq k$,  then by using Lemma~\ref{L:uncorrelated} and the fact that independence and uncorrelatedness are equivalent for Gaussian RVs, the covariance and pseudo-covariance matrices of the received signal vector $\by_k, k=1,\cdots, K,$ can be obtained as
\begin{align}
\bC_{\by_k}&=\xE(\by_k \by_k^H)=\sum_{j=1}^{K}\bH_{kj}\bC_{\bx_j}\bH_{kj}^H+\sigma^2\mathbf{I}, \label{E:Cyk} \\
\xbC_{\by_k}&=\xE(\by_k \by_k^T)=\sum_{j=1}^{K}\bH_{kj}\xbC_{\bx_j}\bH_{kj}^T, \label{E:PCyk}
\end{align}
where in \eqref{E:PCyk}, we have used the fact that the pseudo-covariance of the CSCG noise vector $\mathbf n_k$ is a zero matrix. It is obvious from \eqref{E:Cyk} that $\bC_{\by_k}$ is nonsingular. Then with the augmented covariance matrix $\lbC_{\by_k}$ defined as in \eqref{E:augC} and using the Schur complement \cite{361}, we obtain
 \begin{align}
 \big | \lbC_{\by_k} \big | =& \bigg | \begin{matrix}\bC_{\by_k} & \xbC_{\by_k} \\ \xbC_{\by_k}^* & \bC_{\by_k}^* \end{matrix} \bigg |
 =\big | \bC_{\by_k} \big | \big | \bC_{\by_k}^*-\xbC_{\by_k}^*\bC_{\by_k}^{-1}\xbC_{\by_k} \big |\notag \\
 =&\big |\bC_{\by_k} \big |^2 \big | \mathbf{I}- \bC_{\by_k}^{-1}\xbC_{\by_k}\bC_{\by_k}^{-T}\xbC_{\by_k}^{H}\big | ,
 \end{align}
 where we have used the fact that for an invertible matrix $\mathbf A$,  $\mathbf A^{-T}=(\mathbf A^{-1})^T=(\mathbf A ^T)^{-1}$; and in the last equality, we have used the identities $|\mathbf{A}|=|\mathbf{A}^T|$, $|\mathbf{A}\mathbf{B}|=|\mathbf{A}||\mathbf{B}|$, and the facts that $\bC_{\by_k}$ is Hermitian and $\xbC_{\by_k}$ is symmetric. With the transmitted signals being Gaussian, the received signal $\by_k$ is also Gaussian. Then based on Theorem~\ref{T:entropyGaussian}, the differential entropy of $\by_k\in \mathbb{C}^N$ is given by
\begin{align}
h(\by_k)&=\xlog \big((\pi e)^N \big|\bC_{\by_k} \big|\big) + \frac{1}{2}\xlog \big| \mathbf{I}- \bC_{\by_k}^{-1}\xbC_{\by_k}\bC_{\by_k}^{-T}\xbC_{\by_k}^{H}\big|.\notag
\end{align}
Denote $\bs_k$ as the interference-plus-noise term at receiver $k$, i.e., $
 \bs_k=\sum_{j\neq k} \bH_{kj}\bx_j+\bn_k$. Then the covariance and pseudo-covariance matrices of $\mathbf s_k$ are given by
\begin{align}
\bC_{\bs_k}&=\sum_{j\neq k}\bH_{kj}\bC_{\bx_j}\bH_{kj}^H+\sigma^2\mathbf{I}, \label{E:Csk}\\
\xbC_{\bs_k}&=\sum_{j\neq k}\bH_{kj}\xbC_{\bx_j}\bH_{kj}^T.\label{E:PCsk}
\end{align}
Similarly as for $\by_k$, the differential entropy of $\bs_k$ is
\begin{align}
h(\bs_k)=\xlog\big((\pi e)^N \big | \bC_{\bs_k} \big |\big)+\frac{1}{2}\xlog \big | \mathbf{I}-\bC_{\bs_k}^{-1}\xbC_{\bs_k}\bC_{\bs_k}^{-T}\xbC_{\bs_k}^H \big |.\notag
\end{align}
Under the assumption that interference is treated as Gaussian noise, the achievable rate at receiver $k$ with improper Gaussian signaling can be obtained as
\begin{align}
R_k=&I(\bx_k; \by_k)=h(\by_k)-h(\by_k|\bx_k)=h(\by_k)-h(\bs_k) \notag \\
=&\frac{1}{2} \xlog\frac{\big|\lbC_{\by_k} \big|}{\big|\lbC_{\bs_k} \big|}
=\underbrace{\xlog \frac{\big |\sigma^2\mathbf{I}+\sum_{j=1}^{K}\bH_{kj}\bC_{\bx_j}\bH_{kj}^H \big|}
{\big |\sigma^2\mathbf{I}+\sum_{j\neq k}\bH_{kj}\bC_{\bx_j}\bH_{kj}^H\big|}} _{\triangleq R_{k,\text{proper}}(\{\bC_{\bx_j}\})} \notag \\
&\qquad \qquad \qquad +\frac{1}{2}\xlog\frac{\big| \mathbf{I}- \bC_{\by_k}^{-1}\xbC_{\by_k}\bC_{\by_k}^{-T}\xbC_{\by_k}^{H}\big|}
{\big | \mathbf{I}-\bC_{\bs_k}^{-1}\xbC_{\bs_k}\bC_{\bs_k}^{-T}\xbC_{\bs_k}^H \big |}.\label{E:Rk}
\end{align}
The above equation shows that with improper Gaussian signaling, the achievable rate can be expressed as a summation of two terms. The first term, denoted by $R_{k,\text{proper}}(\{\bC_{\bx_j}\})$, is the rate achievable by the conventional proper Gaussian signaling, which is a function of the transmit covariance matrices only. The second term is a function of both the transmit covariance and pseudo-covariance matrices. By setting $\xbC_{\bx_k}=\mathbf{0}, \forall k$, the second term vanishes and \eqref{E:Rk} reduces to the rate expression for the conventional case of proper Gaussian signaling. The separability of the achievable rate by improper Gaussian signaling provides a general method to improve the achievable rate over the conventional proper Gaussian signaling, i.e., for any given covariance matrices obtained by existing proper Gaussian signaling schemes, the rate can be improved with improper Gaussian signaling by choosing the pseudo-covariance matrices that make the second term in \eqref{E:Rk} strictly positive. It is worth noting that this  property does not exist if we convert the complex-valued system in \eqref{E:yk} to an equivalent real-valued system by doubling the input/output dimensions.

In this paper, we are interested in characterizing  the achievable rate region with improper Gaussian signaling. The achievable rate region for the $K$-user MIMO-IC consists of all the rate-tuples for all users that can be simultaneously achieved under a given set of transmit power constraints for each transmitter, denoted by $P_k,\ k=1,...,K$, i.e.,
\begin{align}\label{E:region}
\mathcal{R}\triangleq \bigcup_{\begin{subarray}{l} \xTr\{\bC_{\bx_k}\}\leq P_k, \\ \lbC_{\bx_k}\succeq \mathbf{0}, \forall k\end{subarray}} \bigg\{(r_1, \cdots, r_K): 0\leq r_k \leq R_k,\forall k\bigg\},
\end{align}
where $\lbC_{\bx_k}$ is the augmented covariance matrix of $\bx_k$ defined in \eqref{E:augC}. The constraint $\lbC_{\bx_k}\succeq \mathbf{0}$ follows from Theorem~\ref{T:validPair}. 
In Sections~\ref{S:problemFormulation}-\ref{S:separate}, we will consider the transmit covariance and pseudo-covariance optimizations for achieving the Pareto boundary of the above rate region for the special case of two-user SISO-IC.
\vspace{-1.5ex}
 \subsection{Widely Linear Precoding}
 In this subsection, we consider the practical problem of how to efficiently generate the transmitted signal $\bx_k$ at each transmitter given any valid pair of covariance matrix  $\bC_{\bx_k}$ and pseudo-covariance matrix  $\xbC_{\bx_k}$, from an information-bearing signal $\bd_k$ that is selected from conventional CSCG (proper Gaussian) codebooks. Without loss of generality, we assume $\bd_k\sim \mathcal{CN}(\mathbf{0}, \mathbf{I})$; thus, we have  
\begin{align}
\bC_{\bd_k}=\mathbf{I},\ \xbC_{\bd_k}=\mathbf{0}, \ k=1,\cdots, K. \label{E:Cdj}
\end{align}
First, consider the conventional linear precoding given by
\begin{align}
\bx_k=\mathbf{U}_k\bd_k, \label{E:strictLinear}
\end{align}
where $\mathbf{U}_k$ is the precoding matrix. Then the pseudo-covariance matrix of $\bx_k$ is given by $\xbC_{\bx_k}=\mathbf{U}_k\xbC_{\bd_k}\mathbf{U}_k^T=\mathbf{0}$. This implies that the conventional linear precoding is not able  to map the proper Gaussian signal $\bd_k$ to the improper transmitted Gaussian signal $\bx_k$.

 Since the augmented covariance matrix defined in \eqref{E:augC} contains both the covariance and pseudo-covariance matrices, a necessary condition for a RV $\mathbf z_k$ to have covariance matrix $\bC_{\bx_k}$ and pseudo-covariance matrices $\xbC_{\bx_k}$ is that its augmented covariance matrix satisfies $\lbC_{\mathbf z_k}=\lbC_{\mathbf x_k}$.  This is ensured by the transformation
 \begin{align}
 \bigg[\begin{matrix} \mathbf z_k \\ \mathbf z_k^* \end{matrix}\bigg]= \lbC_{\bx_k}^{\frac{1}{2}} \bigg[\begin{matrix} \bd_k \\ \bd_k^* \end{matrix}\bigg], \label{E:widerTrans}
 \end{align}
 where $\lbC_{\bx_k}^{\frac{1}{2}}$ denotes the generalized Cholesky factor of the positive semidefinite matrix $\lbC_{\bx_k}$, which is defined by $\lbC_{\bx_k}=\lbC_{\bx_k}^{\frac{1}{2}}(\lbC_{\bx_k}^{\frac{1}{2}})^H$. Since $\lbC_{\mathbf d_k}=\mathbf I$ as given in \eqref{E:Cdj}, it is easy to verify that $\mathbf z_k$ in \eqref{E:widerTrans} satisfies $\lbC_{\mathbf z_k}=\lbC_{\mathbf x_k}$. A common method for finding $\lbC_{\bx_k}^{\frac{1}{2}}$ is via eigenvalue decomposition (EVD). Specifically, let the EVD of $\lbC_{\bx_k}$ be expressed as $\lbC_{\bx_k}=\mathbf U \mathbf D \mathbf U^H$; then $ \lbC_{\bx_k}^{\frac{1}{2}}=\mathbf U \mathbf D^{\frac{1}{2}}$ is obtained. However, it is worth pointing out that the above obtained  $\lbC_{\bx_k}^{\frac{1}{2}}$ cannot satisfy  \eqref{E:widerTrans} in general. This is because the two vectors $\mathbf z_k$ and $\mathbf z_k^*$ in \eqref{E:widerTrans} need to be complex conjugate of each other; therefore, the transformation matrix $\lbC_{\bx_k}^{\frac{1}{2}}$ should be designed  with more care than  the conventional EVD. On the other hand, if we can find one $\lbC_{\bx_k}^{\frac{1}{2}}$ such that it has the following structure:
 \begin{align}\label{E:structure}
\lbC_{\bx_k}^{\frac{1}{2}}=\bigg[\begin{matrix}\mathbf B_1 & \mathbf B_2 \\ \mathbf B_2^* & \mathbf B_1^*
\end{matrix} \bigg],
\end{align}
i.e., the upper-left  (w.r.t. upper-right) block is the complex conjugate of the lower-right  (w.r.t. lower-left) block, then \eqref{E:widerTrans} is equivalent to the following two sets of equations:
\begin{align}
& \mathbf z_k=\mathbf B_1 \mathbf d_k + \mathbf B_2 \mathbf d_k^*, \label{E:z1}\\
& \mathbf z_k^*=\mathbf B_2^* \mathbf d_k + \mathbf B_1^* \mathbf d_k^*.\label{E:z1Conj}
\end{align}
It is easy to verify that the two equations given in \eqref{E:z1} and \eqref{E:z1Conj} are consistent, i.e., \eqref{E:z1Conj} is simply obtained by taking the complex conjugate on both sides of \eqref{E:z1} and vice versa. Therefore, the remaining task is to find one $\lbC_{\bx_k}^{\frac{1}{2}}$ with the structure given by \eqref{E:structure}.  To achieve this end, we define the  following $2M\times 2M$ unitary matrix \cite{262}:
\[
\mathbf{T}\triangleq \frac{1}{\sqrt 2}\bigg[\begin{matrix} \mathbf I_M & i \mathbf I_M \\ \mathbf I_M & -i \mathbf I_M \end{matrix} \bigg], \ \mathbf T \mathbf T ^H =\mathbf T ^H \mathbf T = \mathbf I_{2M}.
\]
For any \emph{real-valued} matrix $\mathbf A=\bigg[\begin{matrix}\mathbf A_{11} & \mathbf A_{12} \\ \mathbf A_{21} & \mathbf A_{22}    \end{matrix} \bigg] \in \mathbb{R}^{2M\times 2M}$, it can be verified that the matrix $\mathbf{T}\mathbf A \mathbf{T}^H$ has the structure given in \eqref{E:structure}, i.e.,
\begin{align}\label{E:specialStruc}
\mathbf{T}\mathbf A \mathbf{T}^H=\bigg[\begin{matrix}\mathbf A_1 & \mathbf A_2 \\ \mathbf A_2^* & \mathbf A_1^*    \end{matrix} \bigg],
\end{align}
with $\mathbf A_1=\frac{1}{2}[(\mathbf A_{11}+\mathbf A_{22})+i(\mathbf A_{21}-\mathbf A_{12})]$ and $\mathbf A_2=\frac{1}{2}[(\mathbf A_{11}-\mathbf A_{22})+i(\mathbf A_{21}+\mathbf A_{12})]$.\vspace{-1ex}
\begin{theorem} \label{T:AEVD}\cite{262}
There exists one form of EVD for the augmented covariance matrix $\lbC_{\bx_k}\in \mathbb{C}^{2M \times 2M}$ defined in \eqref{E:augC}, which is given by
\begin{align}
&\lbC_{\bx_k}= (\mathbf T \mathbf V) \mathbf \Lambda (\mathbf T \mathbf V)^H,\label{E:AEVD}
\end{align}
where $\mathbf V \in \mathbb{R}^{2M \times 2M}$ is a real-valued orthogonal matrix and $\mathbf \Lambda=\mathrm{diag}\{\lambda_1, \lambda_2, \cdots, \lambda_{2M}\}$ contains the eigenvalues of $\lbC_{\bx_k}$.
\end{theorem}
In fact, \eqref{E:AEVD} can be obtained by considering  $\mathbf T^H\lbC_{\bx_k}\mathbf T$, which is a \emph{real-valued} matrix  given by
\[
\mathbf T^H\lbC_{\bx_k}\mathbf T =\left[ \begin{matrix} \Re\{\bC_{\bx_k}+\xbC_{\bx_k}\} & \Im\{-\bC_{\bx_k}+\xbC_{\bx_k}\} \\ \Im\{\bC_{\bx_k}+\xbC_{\bx_k}\}& \Re\{\bC_{\bx_k}-\xbC_{\bx_k}\} \end{matrix}\right].
\]
Furthermore, since $\bC_{\bx_k}$ is Hermitian and $\xbC_{\bx_k}$ is symmetric, it can be verified that the matrix $\mathbf T^H\lbC_{\bx_k}\mathbf T$ is \emph{symmetric} as well. Therefore, its real-valued EVD can be written as
\begin{align}
\mathbf T^H\lbC_{\bx_k}\mathbf T=\mathbf V \mathbf \Lambda  \mathbf V^T.\label{E:RealEVD}
\end{align}
The EVD in \eqref{E:AEVD} is then obtained by applying a unitary transformation $\mathbf T$ to \eqref{E:RealEVD}.

For any given $\lbC_{\bx_k}\succeq \mathbf 0$, all the eigenvalues are nonnegative \cite{262}, i.e., $\lambda_l\geq 0, l=1,\cdots,2M$. Thus \eqref{E:AEVD} can be written as
\begin{align}
\lbC_{\bx_k}=\mathbf T \mathbf V \mathbf \Lambda \mathbf V^H \mathbf T^H=(\mathbf T \mathbf V \mathbf \Lambda^{1/2}\mathbf T^H)(\mathbf T \mathbf V \mathbf \Lambda^{1/2}\mathbf T^H)^H.\notag
\end{align}
Then we have
\begin{align}\label{E:CSqrt}
\lbC_{\bx_k}^{\frac{1}{2}}=\mathbf T (\mathbf V  \mathbf \Lambda ^{1/2} )\mathbf T^H=\bigg[\begin{matrix}\mathbf B_1 & \mathbf B_2 \\ \mathbf B_2^* & \mathbf B_1^*
\end{matrix} \bigg],
\end{align}
where the last equality follows from \eqref{E:specialStruc} and the fact that $\mathbf V  \mathbf \Lambda ^{1/2}$ is a real-valued matrix.

From \eqref{E:widerTrans} and \eqref{E:CSqrt},  it follows that to obtain the transmitted signal vector $\mathbf{x}_k$, which is generally improper with  the  covariance and pseudo-covariance matrices specified by $\lbC_{\bx_k}$, the following precoding needs to be applied to the proper information-bearing signal $\mathbf d_k$:
\begin{align}
\mathbf x_k=\mathbf B_1 \mathbf d_k + \mathbf B_2 \mathbf d_k^*, \label{E:widerTrans2}
\end{align}
where $\mathbf B_1$ and $\mathbf B_2$ are the corresponding blocks in $\mathbf T \mathbf V  \mathbf \Lambda ^{1/2} \mathbf T^H$ as shown in \eqref{E:CSqrt}, with $\mathbf V$ and $\mathbf \Lambda$ obtained by the particular form of  EVD in \eqref{E:AEVD}. Following similar terminologies used in existing literatures on improper signal processing such as \cite{269}, we refer to the precoding given in \eqref{E:widerTrans2} as \emph{widely linear precoding}. Note that if $\mathbf B_2 =\mathbf{0}$, which is the case when  $\lbC_{\bx_k}$ is block-diagonal (i.e., $\xbC_{\bx_k}=\mathbf 0$),  \eqref{E:widerTrans2} reduces to the conventional  linear precoding for proper Gaussian signaling given by \eqref{E:strictLinear}.
 Last, in terms of the real-valued representation, \eqref{E:widerTrans2} can be re-expressed as
\[
\bigg[\begin{matrix} \Re \{\bx_k\} \\ \Im\{\bx_k\} \end{matrix} \bigg]= \bigg[\begin{matrix}\Re\{\mathbf B_1 + \mathbf B_2\} & \Im \{\mathbf B_2 - \mathbf B_1\} \\ \Im\{\mathbf B_2 + \mathbf B_1\} & \Re\{\mathbf B_1 -\mathbf B_2\} \end{matrix}\bigg] \bigg[\begin{matrix} \Re\{\mathbf d_k\} \\ \Im\{\mathbf d_k\}\end{matrix} \bigg]. \notag
\]
\vspace{-1.5ex}
\section{Pareto Boundary Characterization for the Two-user SISO-IC}\label{S:problemFormulation}
In the remaining part of this paper, we will focus on the special  two-user SISO-IC case, with the aim of characterizing its Pareto rate boundary  with improper Gaussian signaling by optimizing both the covariances and pseudo-covariances of transmitted signals. The input-output relationship for the two-user SISO-IC can be simplified from \eqref{E:yk} as
 \begin{equation}\label{A:OSISO}
\begin{aligned}
y_1&=h_{11}x_1+h_{12}x_2+n_1, \\
y_2&=h_{21}x_1+h_{22}x_2+n_2,
\end{aligned}
\end{equation}
where $h_{kj}=|h_{kj}|e^{i \phi_{kj}}, k,j\in \{1,2\}$, is the complex scalar channel from transmitter $j$ to receiver $k$.\footnote{Since a phase rotation can be applied at each of the receivers with coherent demodulation, without loss of generality, the direct channel gains $h_{11}$ and $h_{22}$ can be assumed to be real values. However, this assumption will not change the remaining results in this paper.} Denote the covariance and pseudo-covariance of the transmitted signal $x_k$  by
\begin{align}
C_{x_k}&=\xE(x_kx_k^*), \ \xC_{x_k}=\xE(x_kx_k),\ k=1,2.
\end{align}
Note that $C_{x_k}$'s are nonnegative real numbers equal to the transmit power values of the corresponding users, while $\xC_{x_k}$'s are complex numbers in general. 
Due to Schur complement, the necessary and sufficient conditions stated in Theorem~\ref{T:validPair} for the special case of two-user SISO-IC reduce to
\begin{align}\label{E:validPair}
C_{x_k}\geq 0, \ |\xC_{x_k}|^2\leq C_{x_k}^2, \  k=1,2.
\end{align}
 The covariance and pseudo-covariance of $y_k$ can be written as
\begin{align}
C_{y_k}& =\xE(y_k y_k^*)=|h_{k1}|^2C_{x_1}+|h_{k2}|^2C_{x_2}+\sigma^2, \label{E:y1C}\\
\xC_{y_k}&=\xE(y_k y_k)=h_{k1}^2\xC_{x_1}+h_{k2}^2\xC_{x_2}. \label{E:y1PC}
\end{align}
 For the interference-plus-noise term $s_k=h_{k \bar{k}}x_{\bar{k}}+n_k$, $\bar{k}\neq  k$, we have
\begin{align}
C_{s_k}& =|h_{k \bar{k}}|^2C_{x_{\bar{k}}}+\sigma^2,\
\xC_{s_k}=h_{k \bar{k}}^2\xC_{x_{\bar{k}}}. \label{E:s1PC}
\end{align}
Then the achievable rate expression in \eqref{E:Rk} for the special case of two-user SISO-IC reduces to
\begin{align}
\hspace{-1ex}&R_k^{\text{SISO}}
=\frac{1}{2}\xlog\frac{C_{y_k}^2-|\xC_{y_k}|^2}{C_{s_k}^2-|\xC_{s_k}|^2} \label{E:Rk-SISO-JOINT}\\
&=\underbrace{\xlog \left(1+\frac{|h_{kk}|^2C_{x_k}}{\sigma^2+|h_{k \bar{k}}|^2C_{x_{\bar{k}}}}\right)}_{\triangleq R_{k,\text{proper}}^{\text{SISO}}(C_{x_1}, C_{x_2})}+\frac{1}{2}\xlog\frac{1-C_{y_k}^{-2}|\xC_{y_k}|^2}{1-C_{s_k}^{-2}|\xC_{s_k}|^2}.\label{E:Rk-SISO}
\end{align}
To characterize the Pareto boundary of the achievable rate region defined in \eqref{E:region},  we adopt the rate-profile technique proposed in \cite{240}.  Specifically, any Pareto-optimal rate-pair can be obtained by solving the following optimization problem with a given rate-profile specified  by  {\boldmath$\alpha$}$=(\alpha_1, \alpha_2)$:
\begin{align}
\text{(P1):} & \underset{\{C_{x_k}\}, \{\xC_{x_k}\}, R}{\max.}   R \notag \\
& \qquad \text{s.t.}   \quad R_k^{\text{SISO}}\geq \alpha_k R,  \ \forall k \label{C:jointIneq7} \\
& \qquad \qquad \ 0\leq C_{x_k} \leq P_k, \ \forall k\\
& \qquad \qquad \ |\xC_{x_k}|^2\leq C_{x_k}^2, \ \forall k,
\end{align}
where $\alpha_k$ denotes the target ratio between user $k$'s achievable rate and the users' sum-rate, $R$. Without loss of generality, we assume that $\alpha_1, \alpha_2 > 0$ and $\alpha_1+\alpha_2=1$. 
 Denote the optimal solution to (P1) as $R^{\star}$, then the rate-pair $( \alpha_1 R^{\star}, \alpha_2 R^{\star})$ must be on the Pareto boundary corresponding to the rate-profile given by $(\alpha_1, \alpha_2)$. Thereby, by solving (P1) with different rate-profile parameters $(\alpha_1, \alpha_2)$, the complete Pareto boundary for the achievable rate region  can be found. 
\section{Joint Covariance and Pseudo-Covariance Optimization}\label{S:jointOpt}
 In this section, by applying the SDR technique, we propose an approximate solution to the non-convex problem (P1) where the covariance and pseudo-covariance of the transmitted signals are jointly optimized. The approach of using SDR for solving non-convex quadratically constrained quadratic program (QCQP) has been successfully applied to find high-quality approximate solutions for various problems in communication and signal processing (see e.g. \cite{349} and references therein).  By treating $R$ as a slack variable and  substituting $R_k^{\text{SISO}}$ in \eqref{E:Rk-SISO-JOINT}, (P1) can be equivalently written as
 \begin{align}
 \text{(P1.1):} & \underset{\{C_{x_k}\},\{\xC_{x_k}\}}{\max.} \ \underset{k=1,2}{\min.}\   \frac{1}{2 \alpha_k} \xlog \frac{C_{y_k}^2-|\xC_{y_k}|^2}{C_{s_k}^2-|\xC_{s_k}|^2} \notag \\
&  \qquad \ \text{s.t.} \quad 0\leq C_{x_k}\leq P_k, \quad \forall k \label{C:powerSISO}\\
&\qquad \quad \quad \ |\xC_{x_k}|^2\leq C_{x_k}^2, \quad \forall k.\label{C:validPairSISO}
 \end{align}
(P1.1) is a minimum-weighted-rate maximization (MinWR-Max) problem, where the weights are related to the rate-profile $\boldsymbol \alpha$. The following result will be used for solving (P1.1).
\begin{lemma}\label{L:pos}
For any set of $\{C_{x_k}\}$ and $\{\xC_{x_k}\}$ that is feasible to (P1.1), the following inequalities hold:
\begin{align}
&C_{y_k}^2-|\xC_{y_k}|^2\geq \sigma^4>0, \quad \forall k, \label{E:pos1}\\
&C_{s_k}^2-|\xC_{s_k}|^2\geq \sigma^4>0, \quad \forall k.\label{E:pos2}
\end{align}
\end{lemma}
\begin{IEEEproof}
Please refer to Appendix~\ref{A:pos}.
\end{IEEEproof}
Define the following $2$-dimensional real-valued vectors:
\begin{align}
&\mathbf c = \left[\begin{matrix}C_{x_1} & C_{x_2} \end{matrix} \right]^T,\notag \\
&\mathbf a_1=\left[\begin{matrix}|h_{11}|^2 &  |h_{12}|^2  \end{matrix} \right]^T,\ \mathbf b_1=\left[\begin{matrix}0 & |h_{12}|^2 \end{matrix} \right]^T,  \notag \\
&\mathbf a_2=\left[\begin{matrix}|h_{21}|^2 &  |h_{22}|^2  \end{matrix} \right]^T, \  \mathbf b_2=\left[\begin{matrix}|h_{21}|^2 & 0 \end{matrix} \right]^T.\notag
\end{align}
Then from \eqref{E:y1C} and \eqref{E:s1PC}, we have
\begin{align}
C_{y_k}^2=(\sigma^2+\mathbf a_k^T\mathbf c)^2, \ C_{s_k}^2=(\sigma^2+\mathbf b_k^T\mathbf c)^2, \ k=1,2.\label{E:CykReal}
\end{align}
Define the following $2$-dimensional complex-valued vectors:
\begin{align}
&\mathbf {q}=\left[\begin{matrix}\xC_{x_1} &  \xC_{x_2} \end{matrix} \right]^T,\notag \\
&\mathbf f_1=\left[\begin{matrix}h_{11}^2 & h_{12}^2  \end{matrix} \right]^H,\ \mathbf g_1=\left[\begin{matrix}0 & h_{12}^2\end{matrix} \right]^H,\notag \\
&\mathbf f_2=\left[\begin{matrix}h_{21}^2 & h_{22}^2  \end{matrix} \right]^H,\ \mathbf g_2=\left[\begin{matrix}h_{21}^2 & 0\end{matrix} \right]^H\notag.
\end{align}
 Then from \eqref{E:y1PC} and \eqref{E:s1PC}, we have
\begin{align}
&|\xC_{y_k}|^2=|\mathbf f_k^H \mathbf { q}|^2=\mathbf q^H \mathbf F_k \mathbf q, \label{E:PCykReal}\\
&|\xC_{s_k}|^2=|\mathbf g_k^H \mathbf {q}|^2=\mathbf q^H \mathbf G_k \mathbf q, \label{E:PCskReal}
\end{align}
where $\mathbf F_k =\mathbf f_k \mathbf f_k^H$ and $\mathbf G_k=\mathbf g_k \mathbf g_k^H$, $k=1,2$.
By substituting  \eqref{E:CykReal}-\eqref{E:PCskReal} into (P1.1), 
we obtain the following equivalent problem
\begin{align}
\text{(P1.2):}  \underset{\mathbf c\in \mathbb{R}^2, \mathbf q \in \mathbb{C}^{2}}{\max.}& \underset{k}{\min.}\ \frac{1}{2\alpha_k}\xlog \frac{(\sigma^2+\mathbf a_k^T\mathbf c)^2-\mathbf q^H\mathbf F_k\mathbf q}{(\sigma^2+\mathbf b_k^T\mathbf c)^2-\mathbf q^H\mathbf G_k\mathbf q}\notag \\
 \text{s.t.}\quad
 & \mathbf c^T\mathbf E_k\mathbf c \leq P_k^2, \quad \forall k \label{C:power}\\
 &  \mathbf e_k^T\mathbf c\geq 0,\quad \forall k \label{C:nonNegativePower} \\
&  \mathbf q^H \mathbf{E}_k\mathbf q\leq \mathbf c^T \mathbf{E}_k \mathbf c, \quad \forall k \label{C:pseudo}
\end{align}
where $\mathbf e_k$ is the $k$th column in the identity matrix $\mathbf I_2$,  and $\mathbf E_k=\mathbf e_k\mathbf e_k^T$. \eqref{C:power} and \eqref{C:nonNegativePower} correspond to the constraints \eqref{C:powerSISO} in (P1.1), and \eqref{C:pseudo} is equivalent  to \eqref{C:validPairSISO}.
The objective function of (P1.2) is given by the minimum of weighted log-fraction of quadratic functions over $\mathbf c$ and $\mathbf q$. Due to the noise power $\sigma^2$, the quadratics are non-homogeneous \cite{349}. 
 By introducing a new variable $t$, we obtain the homogenized quadratics \cite{349}, which yield
\begin{align}
\text{(P1.3):} \quad &  \underset{\mathbf c, \mathbf q, t}{\max.} \ \underset{k}{\min.}\ \frac{1}{2\alpha_k}\xlog \frac{(\sigma^2 t+\mathbf a_k^T\mathbf c)^2-\mathbf q^H\mathbf F_k\mathbf q}{(\sigma^2 t+\mathbf b_k^T\mathbf c)^2-\mathbf q^H\mathbf G_k\mathbf q}\notag \\
 & \quad  \text{s.t.} \quad
  \mathbf c^T \mathbf{E}_k \mathbf c \leq P_k^2, \quad \forall k\\
 & \qquad \quad \  \mathbf e_k^T\mathbf c t\geq 0, \quad \forall k \\
 & \qquad \quad \  \mathbf q^H \mathbf{E}_k\mathbf q\leq \mathbf c^T \mathbf{E}_k \mathbf c, \quad \forall k\\
 & \qquad \quad \ t^2=1.
\end{align}
(P1.3) is equivalent to (P1.2) in the sense that if it has an optimal solution $(\mathbf c^{\star}, \mathbf q^{\star}, t^{\star})$, then $(\mathbf c^{\star}/t^{\star}, \mathbf q^{\star}/t^{\star})$ is an optimal solution to (P1.2) with the same optimal value. Therefore, (P1.2) can be solved by solving (P1.3). Next, we show that the celebrated SDR technique can be applied to find an approximate solution to (P1.3). Define
\begin{align}
& \mathbf C=\left[\begin{matrix} t \\ \mathbf c \end{matrix} \right]\left[\begin{matrix}t \\ \mathbf c \end{matrix} \right]^T, \ \mathbf Q=\mathbf q \mathbf q^H,\label{E:psdMatrix}
\end{align}
\begin{align}
& \mathbf A_k=\left[\begin{matrix}\sigma^2 \\ \mathbf a_k \end{matrix} \right]\left[\begin{matrix}\sigma^2 \\ \mathbf a_k \end{matrix} \right]^T, \  \mathbf B_k=\left[\begin{matrix}\sigma^2 \\ \mathbf b_k \end{matrix} \right]\left[\begin{matrix}\sigma^2 \\ \mathbf b_k \end{matrix} \right]^T \\
& \mathbf K_k=\left[\begin{matrix}0& \frac{1}{2}\mathbf e_k^T \\ \frac{1}{2} \mathbf e_k & \mathbf 0\end{matrix} \right], \ \mathbf {\hat{E}}_k=\left[\begin{matrix}0& \mathbf 0\\  \mathbf 0 & \mathbf E_k\end{matrix} \right].
\end{align}
With the identity $\mathbf x^H\mathbf A \mathbf x=\xTr \left( \mathbf A\mathbf x \mathbf x^H\right)$, (P1.3) is recast as
\begin{align}
\text{(P1.4):} & \underset{\mathbf C\in \mathbb{S}^{3}, \mathbf Q\in \mathbb{H}^2}{\max.}  \ \underset{k}{\min.}\ \frac{1}{2\alpha_k}\xlog \frac{\xTr(\mathbf A_k \mathbf C)-\xTr(\mathbf F_k\mathbf Q)}{\xTr(\mathbf B_k \mathbf C)-\xTr(\mathbf G_k\mathbf Q)}\notag \\
& \qquad \ \text{s.t.}  \quad    \xTr(\mathbf{\hat{E}}_k \mathbf C)\leq P_k^2, \quad \forall k \label{C:powerLinear}\\
  & \qquad \quad \quad \ \ \xTr(\mathbf K_k \mathbf C)\geq 0, \quad \forall k \\
 & \qquad \quad \quad \ \ \xTr(\mathbf{E}_k\mathbf Q)\leq \xTr(\mathbf{\hat{E}}_k \mathbf C), \quad \forall k  \\
&  \qquad \quad \quad \ \ \mathbf C_{11}=1 \label{C:C11Linear}\\
&  \qquad \quad \quad \ \ \mathbf C \succeq \mathbf 0, \ \mathbf Q \succeq \mathbf 0 \label{C:positiveSemidefinite}\\
&  \qquad \quad \quad \ \ \mathrm{rank}(\mathbf C)=1, \ \mathrm{rank}(\mathbf Q)=1,\label{C:rank}
\end{align}
where $\mathbf C_{11}$ denotes the $(1,1)$-th entry of $\mathbf C$; the positive semidefinite constraints \eqref{C:positiveSemidefinite} and the rank-1 constraints \eqref{C:rank} are due to \eqref{E:psdMatrix}. With such a reformulation, the objective function of (P1.4) is now a log-fraction of linear functions of $\mathbf C$ and $\mathbf Q$, and all the constraints \eqref{C:powerLinear}-\eqref{C:C11Linear} are also   linear. With Lemma~\ref{L:pos} and the equivalence between (P1.1) and (P1.4), for any pair of $\mathbf C$ and $\mathbf Q$ that is feasible to (P1.4), we have
\begin{align}
&\xTr(\mathbf A_k \mathbf C)-\xTr(\mathbf F_k\mathbf Q)\geq \sigma^4>0, \ \forall k, \label{E:additionalConstraint1}\\
&\xTr(\mathbf B_k \mathbf C)-\xTr(\mathbf G_k\mathbf Q)\geq \sigma^4>0, \ \forall k.\label{E:additionalConstraint2}
\end{align}
The SDR problem of (P1.4) is obtained by discarding the non-convex rank-1 constraints in \eqref{C:rank}, and  including the extra constraints \eqref{E:additionalConstraint1} and \eqref{E:additionalConstraint2}, i.e.,
\begin{align}
\text{(P1.4-SDR):} &\underset{\mathbf C\in \mathbb{S}^{3}, \mathbf Q\in \mathbb{H}^2}{\max.} \underset{k}{\min.}\ \frac{1}{2\alpha_k}\xlog \frac{\xTr(\mathbf A_k \mathbf C)-\xTr(\mathbf F_k\mathbf Q)}{\xTr(\mathbf B_k \mathbf C)-\xTr(\mathbf G_k\mathbf Q)}\notag \\
 & \qquad \ \text{s.t.} \quad
\eqref{C:powerLinear}-\eqref{C:positiveSemidefinite}, \
 \eqref{E:additionalConstraint1},\  \eqref{E:additionalConstraint2}. \notag
\end{align}
Note that although the constraints \eqref{E:additionalConstraint1} and \eqref{E:additionalConstraint2} are redundant in the rank-constrained problem (P1.4), in the rank-relaxed problem (P1.4-SDR), the advantages of including them  are twofold. First, it makes (P1.4-SDR) a problem with less relaxation to (P1.4). Besides, the strict positivity of \eqref{E:additionalConstraint2} makes (P1.4-SDR) a quasi-convex problem and hence can be solved with the standard bisection method \cite{202}. Since any feasible solution of (P1.4) is feasible for (P1.4-SDR), the optimal objective value of (P1.4-SDR) provides an upper bound on that of (P1.4). To solve the quasi-convex problem (P1.4-SDR),  consider the following feasibility problem for a fixed $R$:
\begin{align}
\text{(P1.5):} \ \text{find} & \ \mathbf C\in \mathbb{S}^{3}, \mathbf Q\in \mathbb{H}^2 \notag \\
\text{s.t.} &  \     \eqref{C:powerLinear}-\eqref{C:positiveSemidefinite}, \  \eqref{E:additionalConstraint1}, \ \eqref{E:additionalConstraint2}, \notag \\
&\ \xTr(\mathbf A_k \mathbf C)-\xTr(\mathbf F_k\mathbf Q) \notag \\
& \ \  \geq e^{2\alpha_kR}\left(\xTr(\mathbf B_k \mathbf C)-\xTr(\mathbf G_k\mathbf Q)\right),\ \forall k \label{C:equalConstraint}.
\end{align}
(P1.5) is a SDP problem, which can be efficiently solved  \cite{202}. If (P1.5) is feasible, then the optimal objective value $R_{\text{sdr}}$  of (P1.4-SDR) satisfies $R_{\text{sdr}}\geq R$; otherwise $R_{\text{sdr}}< R$. Therefore, (P1.4-SDR) can be solved by solving the SDP problem (P1.5) together with a bisection search over $R$.

Denote the solution to (P1.4-SDR) by $(\mathbf C^{\star}, \mathbf Q^{\star})$. If $\mathrm{rank}(\mathbf C^{\star})=1$ and $\mathrm{rank}(\mathbf Q^{\star})=1$, then $(\mathbf C^{\star}, \mathbf Q^{\star})$ is also the optimal solution to the rank-constrained problem (P1.4). In this case, SDR is tight and the optimal solution to (P1.3) is given by the principal components of $\mathbf C^{\star}$ and $\mathbf Q^{\star}$, from which the solution to (P1.2) can be obtained; otherwise, we apply the following Gaussian randomization procedure  customized to our problem to obtain an approximate solution to (P1.2)~\cite{349}.
\begin{algorithm}[H]
\caption{Gaussian Randomization Procedure for (P1.2)}
\label{A:randomizationJoint}
\textbf{Input:} The solution $(\mathbf C^{\star}, \mathbf Q^{\star})$ to (P1.4-SDR) and the number of randomization trials $L$.
\begin{algorithmic}[1]
\FOR{$l=1,\cdots, L$}
 \STATE Generate $\bigg[\begin{matrix} t_{l} \\ \boldsymbol\xi_{l} \end{matrix} \bigg]\sim \mathcal{N}(\mathbf 0, \mathbf C^{\star})$, $\boldsymbol\beta_{l}\sim \mathcal{CN}(\mathbf 0, \mathbf Q^{\star})$.
 \STATE Let $\mathbf c_{l}=\boldsymbol\xi_{l}/t_{l}$, and $\mathbf q_{l}=\boldsymbol\beta_{l}/t_{l}$.
  \STATE Construct a feasible point $(\mathbf{\check c}_{l}, \mathbf{\check q}_{l})$ for (P1.2) as follows:
  \begin{align}
   &[\mathbf{\check c}_{l}]_k=\max \left(0, \min\left([\mathbf{c}_{l}]_k, P_k\right)\right), \label{E:projPower}\\
  &[\mathbf{\check q}_{l}]_k=\eta_k [\mathbf q_{l}]_k, \label{E:projPseudo}\\
  &\text{where}  \  \eta_k =\min \left\{1, \frac{[\mathbf{\check c}_{l}]_k}{|[\mathbf{q}_{l}]_k|}\right\},\ k=1,2.\notag
  \end{align}
\ENDFOR
\STATE Let $(\mathbf{\hat c}, \mathbf {\hat q})$ be the solution to
\begin{align} \underset{\mathbf{\check c}_{l}, \mathbf{\check q}_{l}, l=1,\cdots,L}{\max.}
\underset{k}{\min.} \frac{1}{2\alpha_k}\xlog \frac{(\sigma^2+\mathbf a_k^T \mathbf{\check c}_{l})^2-\mathbf {\check q}_l^{H}\mathbf F_k \mathbf{\check q}_{l}}{(\sigma^2+\mathbf b_k^T \mathbf{\check  c}_{l})^2-\mathbf {\check q}_l^{H}\mathbf G_k \mathbf {\check q}_{l}}\notag
\end{align}
\end{algorithmic}
\textbf{Output:} $(\mathbf{\hat c}, \mathbf {\hat q})$ as an approximate solution for (P1.2).
\end{algorithm}
\section{Separate Covariance and Pseudo-Covariance Optimization}\label{S:separate}
For the algorithm proposed in the preceding section, although joint optimizations are performed over the covariances and pseudo-covariances, it is not clear whether a rate gain over
conventional proper Gaussian signaling is attainable since the
obtained solutions are not necessarily globally optimal. In this
section, by utilizing the result that the user's achievable
rate is separable as shown in \eqref{E:Rk-SISO}, we propose a separate
covariance and pseudo-covariance optimization algorithm for
(P1). Specifically, the covariances of the transmitted signals
are first optimized by setting the pseudo-covariances to be
zero, i.e., proper Gaussian signaling is applied. Then,
the pseudo-covariances are optimized with the covariances
fixed as the previously optimized  values. With such a separation
 approach, the obtained improper signaling scheme is guaranteed to improve the rate over proper Gaussian signaling scheme.
\subsection{Covariance Optimization}
When restricted to proper Gaussian signaling with $\xC_{x_1}=0$ and $\xC_{x_2}=0$, by substituting \eqref{E:Rk-SISO} into \eqref{C:jointIneq7}, (P1) reduces to
\begin{align}
\text{(P1.6):}\quad \underset{r,C_{x_1},C_{x_2}}{\max.}  & \quad r \notag \\
\text{s.t.}\quad & \xlog \Big(1+\frac{|h_{11}|^2C_{x_1}}{\sigma^2+|h_{12}|^2C_{x_2}}\Big)\geq  \alpha_1 r,  \notag \\ 
& \xlog \Big(1+\frac{|h_{22}|^2C_{x_2}}{\sigma^2+|h_{21}|^2C_{x_1}}\Big)\geq \alpha_2 r, \notag \\ 
&0\leq C_{x_1} \leq P_1, \ 0\leq C_{x_2}\leq P_2. \notag
\end{align}
For any fixed value $r$, (P1.6) can be transformed to the following feasibility problem:
\begin{align}
\text{(P1.7):} \quad \text{Find} \quad &   C_{x_1}\in \mathbb R,\ C_{x_2} \in \mathbb R \notag \\
\text{s.t.} \quad &   |h_{11}|^2C_{x_1}\geq (\sigma^2+|h_{12}|^2C_{x_2})(e^{\alpha_1 r}-1),  \notag \\
& |h_{22}|^2C_{x_2}\geq (\sigma^2+|h_{21}|^2C_{x_1})(e^{\alpha_2 r}-1), \notag \\
&0\leq C_{x_1} \leq P_1, \ 0\leq C_{x_2}\leq P_2. \notag
\end{align}
(P1.7) is a linear programming (LP) problem, which can be efficiently solved \cite{202}.
 If $r$ is feasible to (P1.7), then the optimal value of (P1.6) satisfies $r^{\star}\geq r$; otherwise, $r^{\star} < r$. Thus, (P1.6) can be efficiently solved by solving (P1.7) together with the bisection method for updating $r$.
\subsection{Pseudo-Covariance Optimization}
Denote the optimal solution to the covariance optimization problem (P1.6) as $\{r^{\star}, C_{x_1}^{\star}, C_{x_2}^{\star}\}$. By fixing the covariances  as $C_{x_1}^{\star}$ and $C_{x_2}^{\star}$,  (P1) is then optimized over the pseudo-covariances $\xC_{x_1}$ and $\xC_{x_2}$. By substituting the first term in the rate expression \eqref{E:Rk-SISO} with $\alpha_k r^{\star}$, the problem for pseudo-covariance optimization is formulated as
\begin{align}
\text{(P1.8):} \quad \underset{\xC_{x_1},\xC_{x_2}, R}{\max.}& \quad    R \notag \\
\text{s.t.}  \quad &   \alpha_1 r^{\star}+\frac{1}{2}\xlog\frac{1-C_{y_1}^{-2}|\xC_{y_1}|^2}{1-C_{s_1}^{-2}|\xC_{s_1}|^2}\geq \alpha_1 R , \notag \\
& \alpha_2 r^{\star}+\frac{1}{2}\xlog\frac{1-C_{y_2}^{-2}|\xC_{y_2}|^2}{1-C_{s_2}^{-2}|\xC_{s_2}|^2}\geq \alpha_2 R,  \notag \\
& |\xC_{x_1}|^2\leq C_{x_1}^{\star 2}, \  |\xC_{x_2}|^2\leq C_{x_2}^{\star 2} \notag,
\end{align}
where $C_{y_1}, C_{s_1}, C_{y_2}$ and $C_{s_2}$ are the corresponding covariance terms with the transmit covariances $C_{x_1}^{\star}$ and $C_{x_2}^{\star}$.
Again, if a given $R$ is achievable for certain pair of $\xC_{x_1}$ and $\xC_{x_2}$, then the optimal value of (P1.8) satisfies $R^{\star}\geq R$; otherwise, $R^{\star} < R$. Therefore, (P1.8) can be solved via solving a set of feasibility problems together with the bisection method. It can be easily observed that $\{\xC_{x_1}=0, \xC_{x_2}=0, R=r^{\star}\}$ is feasible to (P1.8). Therefore,  $R^{\star}\geq r^{\star}$ is satisfied, i.e., with our proposed separate covariance and pseudo-covariance optimizations, the users' sum-rate corresponding to the rate-profile given by $(\alpha_1, \alpha_2)$ with improper Gaussian signaling is guaranteed to be no smaller than that obtained with the optimal proper Gaussian signaling obtained by solving (P1.6).
Therefore, the remaining problem to be solved is the feasibility problem resulting from (P1.8) for a given $R$. By substituting $\xC_{y_k}$ in \eqref{E:y1PC} and $\xC_{s_k}$ in \eqref{E:s1PC}  into (P1.8) and after some  manipulations, the feasibility problem for a given $R$ can be formulated as
\begin{align}
\text{(P1.9):}\  \text{Find} \quad & \xC_{x_1}\in \mathbb C,\ \xC_{x_2}\in \mathbb C \notag \\
\text{s.t.} \quad & a_1|h_{11}^2\xC_{x_1}+h_{12}^2\xC_{x_2}|^2+b_1\leq |\xC_{x_2}|^2,\label{C:ineq1} \\
& a_2|h_{21}^2\xC_{x_1}+h_{22}^2\xC_{x_2}|^2+b_2\leq |\xC_{x_1}|^2, \label{C:ineq2} \\
& |\xC_{x_1}|^2\leq C_{x_1}^{\star 2}, \label{C:ineq3} \\
 & |\xC_{x_2}|^2\leq C_{x_2}^{\star 2},\label{C:ineq4}
\end{align}
where $a_k=\frac{C_{s_k}^2}{\beta_k C_{y_k}^2 |h_{k\bar{k}}|^4}, b_k=\frac{(1-1/\beta_k)C_{s_k}^2}{|h_{k\bar{k}}|^4}, \beta_k=e^{2\alpha_k(R-r^{\star})}$, $k=1,2, \bar{k}\neq k$.
Since the optimal value of (P1.8) satisfies $R^{\star}\geq r^{\star}$, we can assume that $R\geq r^{\star}$ in (P1.9) without loss of optimality. Then it follows that $\beta_k\geq 1,  b_k\geq 0, \forall k$. In the following, we show that (P1.9) can be efficiently solved via solving a finite number of SOCP problems. First, it can be verified that if $\{\xC_{x_1}, \xC_{x_2}\}$ is feasible for (P1.9), then so is $\{\xC_{x_1}e^{i\omega}, \xC_{x_2}e^{i\omega}\}$. Therefore, without loss of generality, we may choose the common phase rotation $\omega$ so that $\xC_{x_1}$ is real and nonnegative.  Denote the magnitude and phase of $\xC_{x_2}$ by $t$ and $\theta$, respectively, i.e., $\xC_{x_2}=te^{i\theta}$. Then for any fixed value of $\theta$, (P1.9) can be transformed into a SOCP feasibility problem given by
\begin{align}
\text{(P1.10):} \quad \text{Find} \quad & \xC_{x_1}\in \mathbb R, \ t\in \mathbb R \notag \\
\text{s.t.} \quad & \bigg\|\begin{matrix} \sqrt{a_1}(h_{11}^2 \xC_{x_1}+h_{12}^2te^{i\theta}) \\ \sqrt{b_1} \end{matrix} \bigg\| \leq t, \notag \\
& \bigg\|\begin{matrix} \sqrt{a_2}(h_{21}^2\xC_{x_1}+h_{22}^2te^{i\theta}) \\ \sqrt{b_2} \end{matrix} \bigg\| \leq \xC_{x_1}, \notag \\
&\xC_{x_1} \leq C_{x_1}^{\star}, \  t\leq C_{x_2}^{\star}. \notag
\end{align}

\begin{theorem}\label{T:thetaRegion}
The feasibility problem (P1.9) can be optimally solved by solving a finite number of  SOCP problems (P1.10), each for a fixed value $\theta$, where $\theta$ can be restricted to the following discrete set:
\begin{align}
\theta \in \{\pi+2(\phi_{11}-\phi_{12}), \pi+2(\phi_{21}-\phi_{22})\}\cup \Theta_{\mathcal{A}}\cup \Theta_{\mathcal{B}}, \notag
\end{align}
where $\Theta_{\mathcal{A}}$ and $\Theta_{\mathcal{B}}$ are the solution sets for $\theta$ to the following two sets of equations with variables $(\theta, t)$ and $(\theta, \xC_{x_1})$, respectively:
 \begin{align}
 \Theta_{\mathcal{A}}:
\begin{cases}\label{E:thetaA}
a_1|h_{11}^2C_{x_1}^{\star}+h_{12}^2te^{i\theta}|^2+b_1=t^2 \\
a_2|h_{21}^2C_{x_1}^{\star}+h_{22}^2te^{i\theta}|^2+b_2=C_{x_1}^{\star 2}\\
\end{cases}
\end{align}
\begin{align}
\Theta_{\mathcal{B}}:
\begin{cases}\label{E:thetaB}
a_1|h_{11}^2\xC_{x_1}+h_{12}^2C_{x_2}^{\star}e^{i\theta}|^2+b_1=C_{x_2}^{\star 2} \\
a_2|h_{21}^2\xC_{x_1}+h_{22}^2C_{x_2}^{\star}e^{i\theta}|^2+b_2=\xC_{x_1}^{ 2}\\
\end{cases}
\end{align}
\end{theorem}
\begin{IEEEproof}
Please refer to Appendix~\ref{A:thetaRegion}.
\end{IEEEproof}
Theorem~\ref{T:thetaRegion} can be intuitively explained as follows. For the feasibility problem (P1.9), if the constraint \eqref{C:ineq1} is more ``restrictive'' than \eqref{C:ineq2}, then $\theta$ should have a value such that the left hand side (LHS) of \eqref{C:ineq1} is minimized. This corresponds to $\theta=\pi+2(\phi_{11}-\phi_{12})$ so that $h_{11}^2\xC_{x_1}$ and $h_{12}^2\xC_{x_2}$ are antiphase. Similar argument for $\theta=\pi+2(\phi_{21}-\phi_{22})$ can be made. On the other hand, if both \eqref{C:ineq1} and \eqref{C:ineq2} are equally ``restrictive'', a feasible solution tends to make both constraints satisfied with equality, as given by \eqref{E:thetaA} and \eqref{E:thetaB}. $\Theta_{\mathcal{A}}$ and $\Theta_{\mathcal{B}}$ correspond to the cases where either the constraint \eqref{C:ineq3} or \eqref{C:ineq4} is active, which can be assumed without loss of generality as shown by Proposition~\ref{Pr:active} in Appendix~\ref{A:thetaRegion}.
The elements in $\Theta_{\mathcal{A}}$ and $\Theta_{\mathcal{B}}$ can be efficiently  obtained by following the steps in Appendix~\ref{A:thetas}.
\begin{figure}
\centering
\includegraphics[width=3in, height=2in]{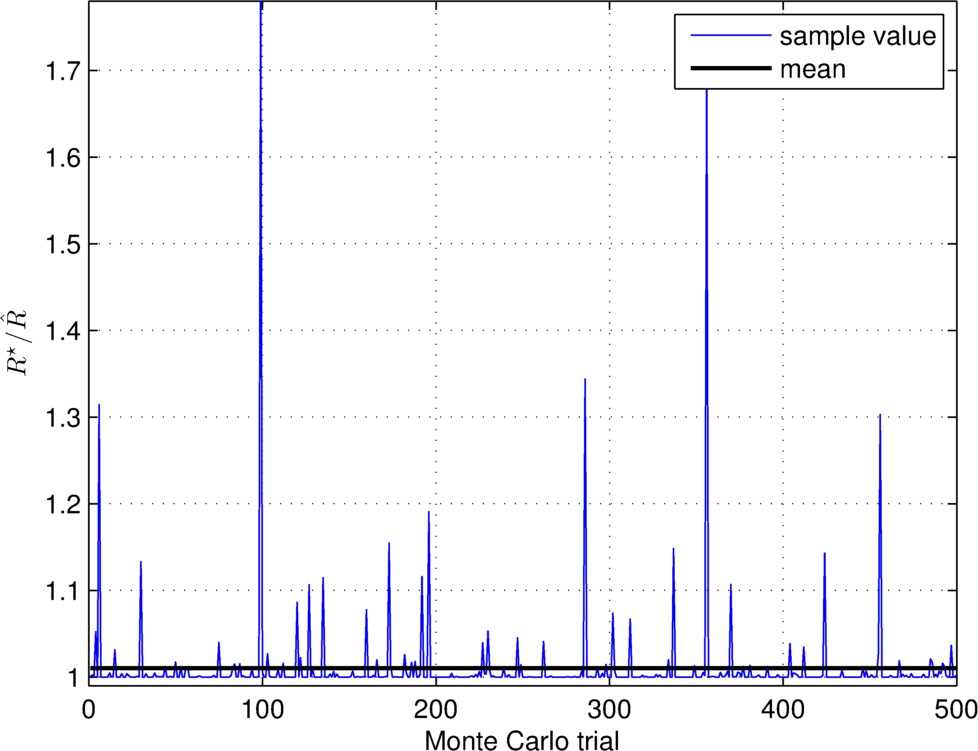}
\caption{Empirical ratio $R^{\star}/\hat R$ for the two-user SISO-IC over $500$ random channel realizations, and with SNR=$0$ dB.}\label{F:ratioRstarOverRhat0dB_K2}
\end{figure}
\section{Numerical Results}\label{S:simulation}
In this section, we evaluate the performance of the proposed algorithms for the two-user SISO-IC with numerical examples. Both transmitters are assumed to have the same power constraint $P$, i.e., $P_1=P_2=P$. SNR is defined as $P/\sigma^2$. For the SDR-based joint covariance and pseudo-covariance optimization algorithm, $L=1000$ is used for the Gaussian randomization procedure in Algorithm~\ref{A:randomizationJoint}.
\begin{figure}
\centering
\includegraphics[width=3in, height=2in]{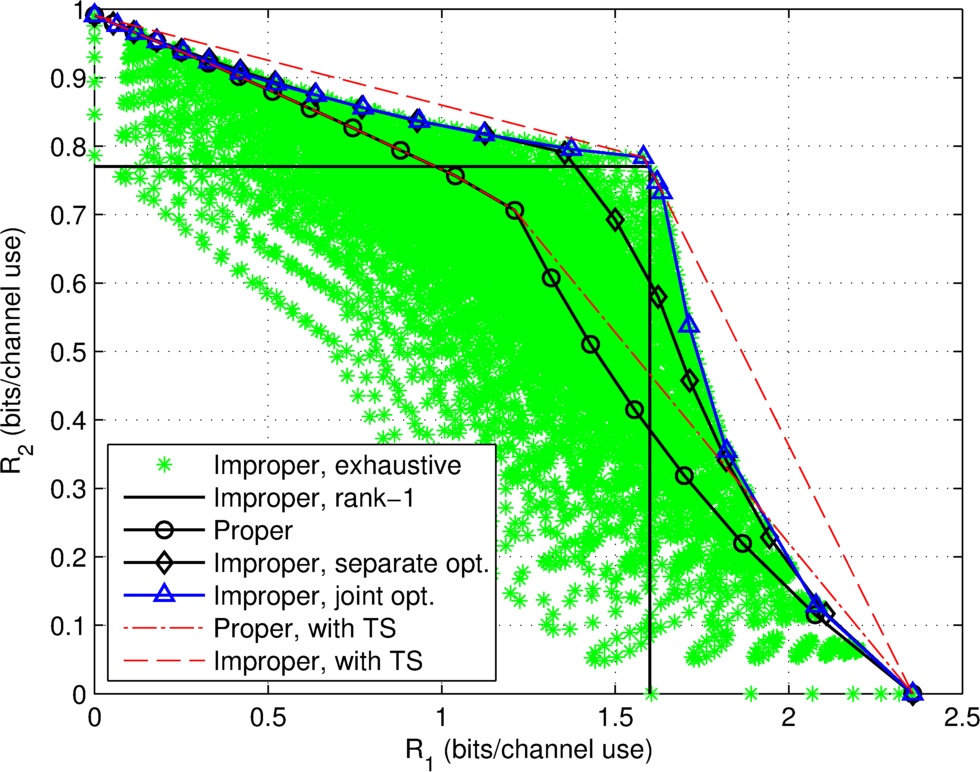}
\caption{Achievable rate region for the two-user SISO-IC with channel realization $\mathbf{H}^{(1)}$, and SNR = 0 dB.\vspace{-2ex}}\label{F:region0dBChannel1}
\end{figure}
\begin{figure}
\centering
\includegraphics[width=3in, height=2in]{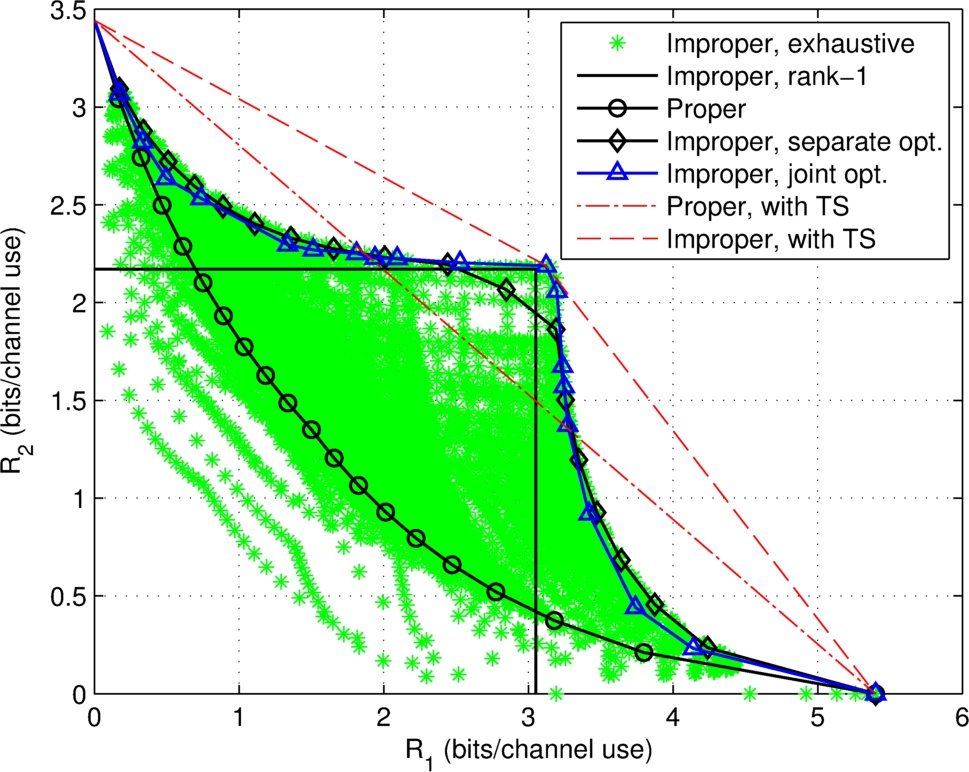}
\caption{Achievable rate region for the two-user SISO-IC with channel realization $\mathbf{H}^{(1)}$, and SNR = 10 dB.\vspace{-2ex}}\label{F:region10dBChannel1}
\end{figure}
\subsection{Approximation Ratio for SDR}
In this subsection, we evaluate the quality of the approximate solution obtained by the SDR-based joint covariance and pseudo-covariance optimization algorithm proposed in Section~\ref{S:jointOpt}. Denote $R^{\star}$  and $R_{\text{sdr}}$ as the optimal objective values of (P1.4) and its relaxation (P1.4-SDR), respectively. Further denote $\hat{R}$ as the objective value of (P1.2) corresponding to the approximate solution obtained by Algorithm~\ref{A:randomizationJoint}. Then
\begin{align}
\hat{R}\leq R^{\star}\leq R_{\text{sdr}}, \text{ or } 1\leq R^{\star}/\hat{R}\leq R_{\text{sdr}}/\hat{R},
\end{align}
where $R^{\star}/\hat{R}$ is the approximation ratio. Since in general  the optimal value $R^{\star}$ is difficult to be found, the upper bound $R_{\text{sdr}}/{\hat R}$ of the approximation ratio is usually used to evaluate the quality of the obtained approximate solutions~\cite{349}. Fortunately, for the two-user SISO-IC considered herein, the optimal value $R^{\star}$ of (P1.1) and hence that of its equivalent problem (P1.4)) can be obtained by the exhaustive search method proposed in \cite{256}.      
With the rate-profile {\boldmath$\alpha$} setting to $(1/2, 1/2)$, the empirical ratios of $R^{\star}/\hat R$ over $500$ random channel realizations at SNR= $0$ dB are plotted in  Fig.~\ref{F:ratioRstarOverRhat0dB_K2}, where the channel coefficients are generated from i.i.d. CSCG random variables with zero-mean and unit-variance. It is found that the mean of $R^{\star}/\hat R$ is $1.01$, which demonstrates the  high quality of the approximate solution obtained by the SDR-based joint covariance and pseudo-covariance optimization algorithm.
\subsection{Rate Region Comparison}
In Fig.~\ref{F:region0dBChannel1} and Fig.~\ref{F:region10dBChannel1}, the achievable rate regions for an example two-user SISO-IC are plotted for SNR=$0$ dB and  $10$ dB, respectively. The channel matrix for both plots is given by $\mathbf{H}^{(1)}=\left[\begin{matrix} h_{11} & h_{12}\\
  h_{21}  & h_{22} \end{matrix} \right]=\left[\begin{matrix} 2.0310e^{-i0.6858} & 1.4766e^{i 2.6452}\\
     0.7280e^{i 1.9726} &  0.9935 e^{-i 0.6676}\end{matrix} \right]$. The proposed improper Gaussian signaling schemes with joint and separate  covariance and pseudo-covariance optimizations  are compared with other existing schemes, including the optimal proper Gaussian signaling scheme by solving (P1.6), the  optimal improper Gaussian signaling obtained by the exhaustive search method \cite{256}, and the  rank-1 improper Gaussian signaling scheme \cite{255}. 
   Both figures reveal that for the given channel $\mathbf{H}^{(1)}$,  the achievable rate regions are significantly enlarged with improper Gaussian signaling over the conventional proper Gaussian signaling. The plots also demonstrate that the SDR-based joint covariance and pseudo-covariance optimization algorithm yields almost the  optimal rates given by the exhaustive search, which is consistent with the observation in Fig.~\ref{F:ratioRstarOverRhat0dB_K2}. Moreover, it is observed that the separate covariance and pseudo-covariance optimization algorithm performs close to the optimal solution, and also always outperforms the optimal proper signaling. It is worth remarking that, even with time-sharing (TS),\footnote{The achievable rate region with TS is obtained by taking the convex-hull operation over all the achievable rate-pairs given in \eqref{E:region}.} improper Gaussian signaling still outperforms proper Gaussian signaling, as shown by the dashed lines in the two figures. For this particular channel realization, the Pareto boundary points of the achievable rate region with TS using improper Gaussian signaling can be obtained by the TS between the two single-user maximum rate points, and the largest rate corner point by the existing rank-1 scheme \cite{255}. However, this is not always the case, as illustrated by the next example.

Next, consider a two-user SISO-IC channel given by $\mathbf{H}^{(2)}=\left[\begin{matrix} 4.0e^{i 1.7730} &  0.90e^{i  1.6744}\\
   0.80e^{i  0.6249}  &  1.50e^{i  2.1057} \end{matrix} \right]$. 
   It is observed from Fig.~\ref{F:region0dBChannel2} that for this particular channel realization at SNR=$0$ dB, there is no notable  performance gain by using improper Gaussian signaling over proper signaling, which is in contrast to that observed in Fig.~\ref{F:region0dBChannel1} with channel $\mathbf H^{(1)}$. This is mainly due to the relatively weaker interfering link in this channel setup. For example, the interference-to-signal power gain ratio at user 1's receiver is given by $|h_{12}|^2/|h_{11}|^2=0.051$, which is much smaller than $0.53$  in channel $\mathbf H^{(1)}$.  Note that intuitively, it is the non-negligible mutual interference among the users that is exploited by improper Gaussian signaling to outperform the conventional proper Gaussian signaling.\footnote{Consider the extreme case where all the interfering link gains vanish to zero and the SISO-IC reduces to  $K$ decoupled  Gaussian point-to-point channels. In this case, proper Gaussian signaling is known to be optimal.} Therefore, for channel realization $\mathbf H^{(2)}$ with almost negligible interfering link for user $1$, no observable  performance gain can be achieved by improper Gaussian signaling.
    Another observation from Fig.~\ref{F:region0dBChannel2} is that the rank-1 improper signaling scheme \cite{255}, which is based on the equivalent real-valued $2\times 2$ MIMO-IC,  gives strictly smaller rate region than that by proper Gaussian signaling. In contrast, our proposed improper signaling schemes with either joint or separate covariance and pseudo-covariance optimizations, are observed  to perform no worse than the optimal proper Gaussian signaling, in accordance with our previous discussion. 
\begin{figure}
\centering
\includegraphics[width=3in, height=2in]{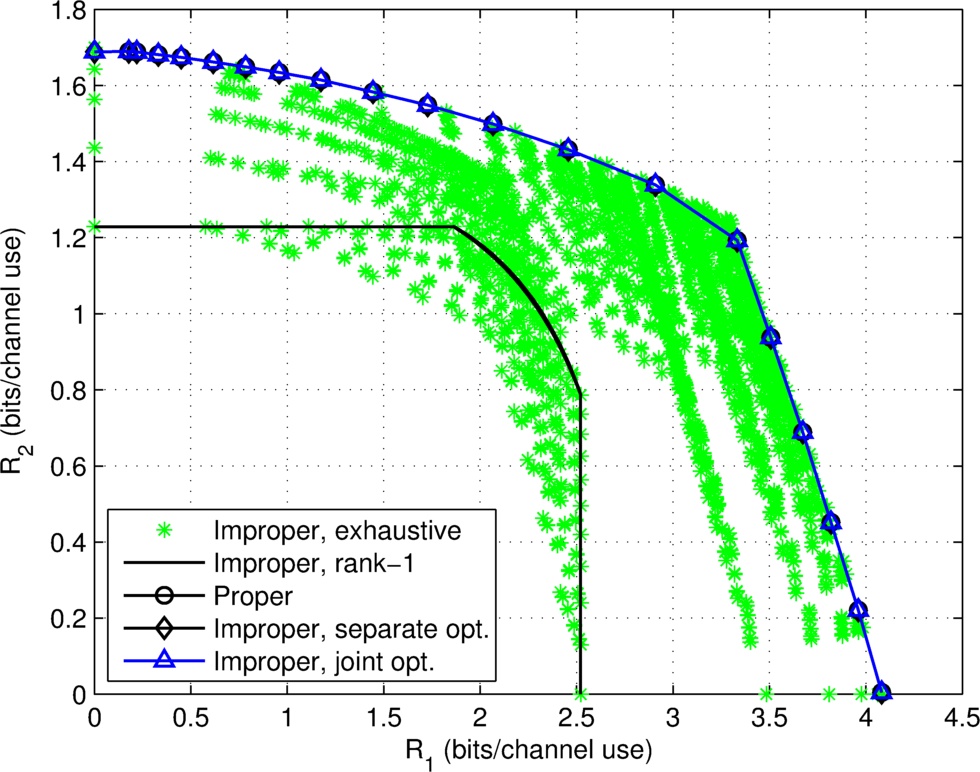}
\caption{Achievable rate region for the two-user SISO-IC with channel realization $\mathbf{H}^{(2)}$, and SNR = 0 dB. \vspace{-2ex}}\label{F:region0dBChannel2}
\end{figure}
\subsection{Max-Min Rate Comparison}
The rate-profile technique used in characterizing the Pareto boundary  of the achievable rate region can be directly applied for  maximizing the minimum (max-min) rate of the two users without TS. Specifically, the max-min problem for the two-user SISO-IC is equivalent to solving (P1) by using the rate-profile {\boldmath$\alpha$} $= (1/2, 1/2)$.  An alternative  max-min solution with improper Gaussian signaling was proposed in \cite{256}, where based on the equivalent real-valued $2\times 2$ MIMO-IC, the  transmit covariance matrix of the equivalent real-valued signal vector for each user is assumed to be of rank-1. Note that the use of rank-1 transmission in both \cite{255} and \cite{256} can be justified by the fact that the total DoF of two-user $2\times 2$ MIMO-ICs exactly equals to $2$ \cite{93}. For the ease of precoder design, zero-forcing (ZF) receivers were further  applied in \cite{256}. As a benchmark comparison, we also plot the max-min rate achievable by the simple time division multiple access (TDMA) scheme, where for simplicity, each user is assumed to access the channel for half of the time.

To evaluate the average max-min rates, $500$ random channel realizations are simulated, where the channel coefficients are drawn from independent zero-mean CSCG random variables. For this example, asymmetric channels are considered, where the average power values of the direct and interfering channels are $1$ and $0.2$, respectively, i.e., $h_{kk}\sim \mathcal{CN} (0, 1)$, $h_{k\bar k}\sim \mathcal{CN}(0, 0.2)$, $k=1,2$, $\bar k \neq k$. The obtained results are shown  in Fig.~\ref{F:maxMin2UserSISO}.  The optimal max-min rate achievable by proper Gaussian signaling and that by improper Gaussian signaling obtained by the exhaustive search method \cite{256} are also included in the figure. It is observed that in the low SNR regime, there is no notable gain by improper Gaussian signaling over conventional proper Gaussian signaling, which is due to the negligible  interference levels at low SNRs.  As SNR increases, the max-min rate by proper Gaussian signaling saturates since the total number of data streams transmitted, which is  $2$, exceeds the total number of DoF of the two-user SISO-IC, which is $1$. In contrast, the linear increase of the max-min rates with respect to the logarithm of SNR can be achieved either by TDMA, or by improper Gaussian signaling. It is worth remarking that over the entire SNR range, the proposed algorithms based on covariance and pseudo-covariance optimizations yield close-to-optimal performance obtained by exhaustive search method.  On the other hand, the rank-1 transmission with ZF receivers based on the equivalent real-valued MIMO-IC gives a near-optimal performance in the high-SNR regime, which is expected due to the optimality of ZF receivers at high SNR as well as the DoF optimality of rank-1 transmission as pointed out  in \cite{255,256}; however, in the low and moderate SNR regime, the rank-1 transmission scheme results in strictly suboptimal performance, which may be  due to the noise enhancement issue  associated with ZF receivers applied in \cite{256}, as well as the over-conservative number of data streams used by assuming rank-1 transmit covariance matrices.
\vspace{-1.5ex}
\begin{figure}
\centering
\includegraphics[width=3in, height=2in]{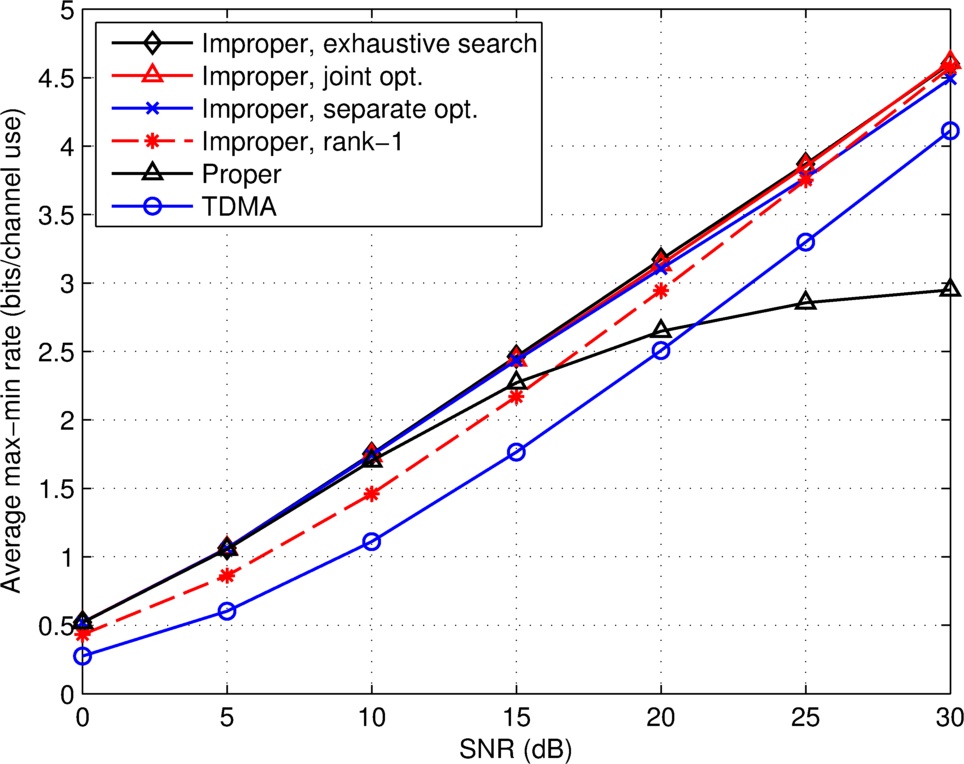}
\caption{Average max-min rate for the two-user SISO-IC.\vspace{-3ex}}\label{F:maxMin2UserSISO}
\end{figure}
\subsection{Sum-Rate Comparison}
In this subsection, the sum-rate maximization with improper Gaussian signaling is considered. By using  the equivalent real-valued MIMO-IC of the complex-valued SISO-IC, existing sum-rate maximization algorithms in the literature, such as the one via the iterative weighted MSE minimization (WMMSE) \cite{335}, can be applied directly  for maximizing the sum-rate of the two-user SISO-IC when improper Gaussian signaling is employed. However, although the WMMSE algorithm is guaranteed to converge to a local maximum of the sum-rate, it is not guaranteed to achieve the global sum-rate maximum. With the algorithms proposed in this paper via covariance and pseudo-covariance optimization, we illustrate with the following example  that our proposed algorithms strictly improve the achievable sum-rate over that by the WMMSE algorithm.

 In order to apply the WMMSE algorithm \cite{335} to the sum-rate maximization problem when improper Gaussian signaling is applied, we  transform the complex-valued channel to the equivalent real-valued MIMO channel, similarly as in \cite{189,255,256}. 
  Denote $\mathbf Q_k$ as the transmit covariance matrix of user $k$ in the equivalent real-valued $2\times 2$ MIMO-IC.  Without loss of generality, denote the rate-pair obtained by the iterative WMMSE scheme by $r\cdot(\delta_1,\delta_2)$, with $\delta_1\geq 0, \delta_2\geq 0$, and $\delta_1+\delta_2=1$.  With the rate-profile {\boldmath$\alpha$}$=(\delta_1,\delta_2)$, (P1) is solved to obtain a new sum-rate $R$. If $R>r$, then the rate-pair obtained by the WMMSE algorithm cannot be sum-rate optimal, and a strictly improved sum-rate can  be obtained by our proposed algorithms. Table~\ref{table2} shows the obtained sum-rate result.  The equivalence of the channel matrices and the converged transmit parameters between the original complex-valued SISO-IC and  the equivalent real-valued MIMO-IC for this example is shown in Table~\ref{table1}, with the SNR set as 10 dB.
  \begin{table*}
\caption{Equivalence  between complex-valued SISO-IC and real-valued MIMO-IC}
\centering
\begin{tabular}{|l | l |l|}
  \hline
    & Complex-valued SISO-IC & Real-valued MIMO-IC \\
   \hline
  \multirow{2}{*}{Channel} & $h_{11}=2.7388 - 0.2498i, h_{12}=0.9956 + 1.8047i$ & $\mathbf{H}_{11}=\bigg[\begin{matrix} 2.7388  &  0.2498 \\
   -0.2498  &  2.7388 \end{matrix} \bigg] , \mathbf{H}_{12}=\bigg [ \begin{matrix} 0.9956  & -1.8047\\
    1.8047  &  0.9956 \end{matrix} \bigg]$ \\
  & $h_{21}=0.6680 - 1.6470i, h_{22}=0.4760 + 1.2706i$ &  $\mathbf H_{21}=\bigg[\begin{matrix} 0.6680  &  1.6470\\
   -1.6470  &  0.6680 \end{matrix} \bigg],  \mathbf H_{22}=\bigg[\begin{matrix} 0.4760 &  -1.2706\\
    1.2706   & 0.4760 \end{matrix} \bigg]$ \\
  \hline
  \multirow{2}{*}{WMMSE Init.} & $C_{x_1}=10.000, \quad \widetilde C_{x_1}=9.546e^{i 0.5512}$ & \multirow{2}{*}{$\mathbf{Q}_1=\bigg [ \begin{matrix} 9.0660   & 2.4998\\
      2.4998   & 0.9340 \end{matrix} \bigg], \mathbf{Q}_2=\bigg [ \begin{matrix} 1.3051 &   0.9125\\
0.9125  &  8.6949 \end{matrix} \bigg]$ }\\
      & $C_{x_2}=10.000, \quad \widetilde C_{x_2}=7.6118 e^{i 2.8995}$ & \\
\hline
\multirow{2}{*}{WMMSE Sol.} & $C_{x_1}=9.9981,  \quad \widetilde C_{x_1}=9.9981e^{
  i 0.4575}$ & \multirow{2}{*}{$\mathbf{Q}_1=\bigg [ \begin{matrix} 9.4840    &  2.2081\\
    2.2081   & 0.5141\end{matrix} \bigg], \mathbf{Q}_2=\bigg [ \begin{matrix} 0.0047 &  -0.2174\\
   -0.2174  &  9.9752 \end{matrix} \bigg]$} \\
      & $C_{x_2}=9.9800, \quad \widetilde C_{x_2}=9.9800 e^{-i 3.0980}$ & \\
\hline
\multirow{2}{*}{Separate opt.} & $C_{x_1}=8.7366, \quad \widetilde C_{x_1}=8.7366$ & \multirow{2}{*}{$\mathbf{Q}_1=\bigg [ \begin{matrix} 8.7366   &      0\\
         0   &      0\end{matrix} \bigg], \mathbf{Q}_2=\bigg [ \begin{matrix} 9.9881  &  0.0708\\
    0.0708   & 0.0006 \end{matrix} \bigg]$} \\
      & $C_{x_2}=9.9887, \quad \widetilde C_{x_2}=9.9885 e^{i 0.0142}$ &\\
    \hline
\multirow{2}{*}{Joint opt.} & $C_{x_1}=10.00, \quad \widetilde C_{x_1}=10.00e^{i1.1204}$ & \multirow{2}{*}{$\mathbf{Q}_1=\bigg [ \begin{matrix} 7.1768 & 4.5013\\
         4.5013 &  2.8232\end{matrix} \bigg], \mathbf{Q}_2=\bigg [ \begin{matrix} 7.5137 &  4.3221\\
   4.3221  &  2.4863 \end{matrix} \bigg]$} \\
      & $C_{x_2}=10.00, \quad \widetilde C_{x_2}=10.00e^{-i 1.0440}$ &\\
  \hline
\end{tabular}
\label{table1}
\end{table*}
\begin{table*}
\caption{Sum-Rate comparison}
\centering
\begin{tabular}{|l | c |c |c|}
\hline
 & WMMSE & Separate opt. & Joint opt. \\
\hline
Rate-pair & $(2.8673,1.8102)$ & $(3.4079,2.1515)$ & $(3.4761,2.2078 )$ \\
\hline
Sum-rate & 4.6775 & 5.5594 & 5.6839 \\
\hline
Improvement & -- & $18.85\%$
 & $21.52\%$ \\
\hline
\end{tabular}
\label{table2}
\end{table*}
 \section{Conclusion}\label{S:conclusions}
This paper studied the transmit optimization for Gaussian ICs when improper or circularly asymmetric complex Gaussian signaling is applied. Under the assumption that the interference is treated as additive Gaussian noise, it was shown that the use of conventional  proper or circularly symmetric complex Gaussian signaling may result in undesired rate loss. A new  achievable rate expression for the general MIMO-IC was derived, which is expressed as a summation of the rate achievable with the conventional proper Gaussian signaling, and an additional  term due to the use of improper Gaussian signaling. This result provides a useful  method to improve the rate over the conventional proper Gaussian signaling  by separately optimizing the covariance and pseudo-covariance matrices.  We also proposed the technique  of widely linear precoding, which efficiently maps the proper Gaussian information-bearing signals to the improper Gaussian transmitted signals with any given pair of covariance and pseudo-covariance matrices. Furthermore, for the two-user SISO-IC, we formulated the optimization problem to characterize the Pareto boundary of the achievable rate region via the rate-profile method. Both joint and separate covariance and pseudo-covariance optimization algorithms were proposed, both of which outperform the conventional proper Gaussian signaling and provide advantages over existing improper Gaussian signaling schemes.
 \appendices
 \section{Proof of Lemma~\ref{L:pos}}\label{A:pos}
The derivations  for \eqref{E:pos1} and \eqref{E:pos2} follow similar arguments. For brevity, we only show that of \eqref{E:pos1} as follows:
\begin{align}
|\xC_{y_k}|^2+\sigma^4&\leq \big(|\xC_{y_k}|+\sigma^2\big)^2
=\Big(\big|h_{k1}^2\xC_{x_1}+h_{k2}^2\xC_{x_2}\big|+\sigma^2\Big)^2 \notag \\
&\overset{(a)}{\leq} \Big(|h_{k1}|^2 |\xC_{x_1}|+|h_{k2}|^2 |\xC_{x_2}|  +\sigma^2 \Big)^2 \notag \\
& \overset{(b)}{\leq} \Big(|h_{k1}|^2 C_{x_1}+|h_{k2}|^2 C_{x_2}  +\sigma^2  \Big)^2 = C_{y_k}^2, \notag
\end{align}
where $(a)$ follows from the triangle inequality, and $(b)$ is true due to the constraint \eqref{C:validPairSISO} in (P1.1).

\section{Proof of Theorem~\ref{T:thetaRegion}}\label{A:thetaRegion}
For notational convenience, in this appendix, we use $X_1$ and $X_2$ to represent $\xC_{x_1}$ and $\xC_{x_2}$, respectively. First, the following proposition shows that to solve (P1.9), we may consider exterior solutions only, i.e., the solutions at which at least one of the inequality constraints is active.
\begin{proposition}\label{Pr:active}
If $\{X_1,X_2\}$ is feasible to (P1.9) with $|X_1|<C_{x_1}^{\star}$ and $|X_2|<C_{x_2}^{\star}$, then there exists another feasible solution $\{X_1', X_2'\}$ with $|X_1'|=C_{x_1}^{\star}$ or $|X_2'|=C_{x_2}^{\star}$.
\end{proposition}
\begin{proof}
Let $\tau\triangleq \min\big\{\frac{C_{x_1}^{\star}}{|X_1|}, \frac{C_{x_2}^{\star}}{|X_2|}\big\}$. Then $\tau >1$. Define $X_1'=\tau X_1$ and $X_2' =\tau X_2$. Then the constraints in \eqref{C:ineq3} and \eqref{C:ineq4} are satisfied by $X_1'$ and $X_2'$, i.e., $|X_1'|\leq C_{x_1}^{\star}$ and $|X_2'|\leq C_{x_2}^{\star}$. Furthermore, at least one of them is satisfied with equality. The constraint in \eqref{C:ineq1} is also satisfied since
\begin{align}
&a_1|h_{11}^2X_1'+h_{12}^2X_2'|^2+b_1
=\tau^2a_1|h_{11}^2X_1+h_{12}^2X_2|^2+b_1 \notag \\
&\overset{(a)}{\leq} \tau^2 (a_1|h_{11}^2X_1+h_{12}^2X_2|^2+b_1)\overset{(b)}{\leq}  \tau^2|X_2|^2=|X_2'|^2, \notag
\end{align}
where $(a)$ comes from  $\tau>1$ and $b_1\geq 0$, and $(b)$ is true since $\{X_1, X_2\}$ is feasible to (P1.9).
Similarly, \eqref{C:ineq2} is also satisfied.  Therefore, $\{X_1', X_2'\}$ is a feasible solution to (P1.9) with at least one of the inequality constraints being active.
\end{proof}

  Next, we derive Theorem~\ref{T:thetaRegion} using the Karush-Kuhn-Tucker (KKT) conditions, which are necessary optimality conditions for the constrained optimization problem (P1.9) \cite{202}. For notational convenience, denote the inequality constraints \eqref{C:ineq1}--\eqref{C:ineq4} by $f_1\leq 0, f_2\leq 0, h_1\leq 0$ and $ h_2\leq 0$, respectively.  Denote $\lambda_1, \lambda_2, \mu_1, \mu_2$ as the corresponding dual variables, respectively. The Lagrangian function of (P1.9) is given by
  \begin{align}
\hspace{-0.5ex}L&(X_1,  X_2, \lambda_1, \lambda_2, \mu_1, \mu_2)=\lambda_1 \Big\{a_1 \big(|h_{11}|^4X_1^*X_1+|h_{12}|^4X_2^*X_2 \notag \\
 &+2\Re \{h_{11}^{2*}h_{12}^2X_1^*X_2\}\big) +b_1-X_2^*X_2\Big\}
+\lambda_2\Big\{a_2 \big(|h_{21}|^4X_1^*X_1 \notag \\
&+|h_{22}|^4X_2^*X_2
+2\Re \{h_{21}^{2*}h_{22}^2X_1^*X_2\}\big)
+b_2-X_1^*X_1\Big\} \notag \\
&+\mu_1(X_1^*X_1-C_{x_1}^{\star 2}) +\mu_2(X_2^*X_2-C_{x_2}^{\star 2}). \label{E:Lagrangian}
\end{align}
For $\{X_1, X_2\}$ to be a solution to (P1.9), the following KKT conditions must be satisfied: 
\begin{enumerate}
\item{Dual feasibility: $\lambda_1\geq 0, \  \lambda_2\geq 0, \ \mu_1\geq 0, \  \mu_2\geq 0.$
}
\item{Zero derivative: The derivatives of the Lagrangian function \eqref{E:Lagrangian} are zero\cite{251}:
    \begin{align}
    \frac{\partial L}{\partial X_2^*}=0 \Rightarrow & -X_2 \underbrace{\big [ (a_1|h_{12}|^4-1)\lambda_1+\lambda_2a_2|h_{22}|^4+\mu_2 \big]}_{\triangleq c_2}\notag \\ &=X_1(\lambda_1\underbrace {a_1h_{12}^{2*}h_{11}^2}_{\triangleq V_1}+\lambda_2\underbrace{a_2h_{22}^{2*}h_{21}^2}_{\triangleq V_2})\notag
    \end{align}
    \begin{align}
\frac{\partial L}{\partial X_1^*}=0 \Rightarrow & X_1 \underbrace{\big [ (a_2|h_{21}|^4-1)\lambda_2+\lambda_1a_1|h_{11}|^4+\mu_1 \big]}_{\triangleq c_1}\notag \\ &=-X_2(\lambda_1\underbrace {a_1h_{11}^{2*}h_{12}^2}_{=V_1^*}+\lambda_2\underbrace{a_2h_{21}^{2*}h_{22}^2}_{=V_2^*}),\notag \\
& \Downarrow \notag \\
 -X_2c_2&=X_1(\lambda_1V_1+\lambda_2V_2), \label{E:X2X1}\\
 X_1c_1&=-X_2(\lambda_1V_1^*+\lambda_2V_2^*). \label{E:X2X1Another}
    \end{align}
    }
    \item{Complementary slackness: $\lambda_1f_1= 0,  \lambda_2 f_2= 0, \mu_1 h_1= 0, \mu_2 h_2= 0.$ 
        }
\end{enumerate}
As discussed previously, without loss of generality, $X_1$ can be assumed to be a nonnegative real number. The following cases are then considered to derive the possible phases $\theta$ of $X_2$:
\begin{itemize}
\item{Case I: $f_1=0$ and $f_2\neq 0$. Then from the complementary slackness condition, we have $\lambda_1>0$ and $\lambda_2=0$. Substituting them into \eqref{E:X2X1Another}, we have
    \[ X_2=-\frac{X_1(\lambda_1a_1|h_{11}|^4+\mu_1)}{\lambda_1|V_1|^2}V_1.\]
     Since $\lambda_1\geq 0$, $\mu_1\geq 0$, $a_1\geq 0$ and $X_1\geq 0$, the phase $\theta$ of $X_2$ equals to that of $V_1$ rotated by $\pi$, which is $\pi+2(\phi_{11}-\phi_{12})$ since $V_1=a_1h_{12}^{2*}h_{11}^2$.}
\item{Case II: $f_1\neq 0$ and $f_2= 0$. Then $\lambda_1=0$ and $\lambda_2>0$. 
     Similarly, by using \eqref{E:X2X1}, 
     we have $\theta=\pi+2(\phi_{21}-\phi_{22})$.}
\item{Case III: $f_1\neq 0$ and $f_2 \neq 0$, then $\lambda_1=0$ and $\lambda_2=0$. By substituting them into \eqref{E:X2X1} and \eqref{E:X2X1Another}, we have $\mu_2X_2=0$ and $\mu_1X_1=0$. If $X_2=0$, then $\theta$ can be any arbitrary value. If $X_1=0$, then $X_2=0$ is implied due to \eqref{C:ineq2} and again, $\theta$ is arbitrary.  Therefore, we may focus on $\mu_1=0$ and $\mu_2=0$ only. However, Proposition~\ref{Pr:active} suggests that we may consider the exterior solutions only, i.e., either $h_1=0$ or $h_2=0$ is satisfied. Thus, $\mu_1>0$ or $\mu_2>0$ can be assumed. Therefore, case III can be ignored without loss of optimality.}
\item{Case IV: $f_1=0$ and $f_2=0$, then $\theta$ belongs to the solution set for the equations given in \eqref{E:thetaA} and \eqref{E:thetaB}, which are obtained by satisfying the constraints $f_1$ in \eqref{C:ineq1} and $f_2$ in \eqref{C:ineq2} with equality. \eqref{E:thetaA} and \eqref{E:thetaB} correspond to $|X_1|=C_{x_1}^{\star}$ and $|X_2|=C_{x_2}^{\star}$, respectively, which can be assumed without loss of generality due to Proposition~\ref{Pr:active}.}
\end{itemize}
This completes the proof of Theorem~\ref{T:thetaRegion}.
\section{Solving $\Theta_{\mathcal{A}}$ and $\Theta_{\mathcal{B}}$ in Theorem~\ref{T:thetaRegion}}\label{A:thetas}
 In this appendix, we show the steps to solve  $\Theta_{\mathcal{A}}$, while the solution of $\Theta_{\mathcal{B}}$ can be obtained similarly and thus omitted. Note that the unknown variables in \eqref{E:thetaA} are $\theta$ and $t$. 
 After some manipulations, \eqref{E:thetaA} can be written as
\begin{align}
 t &\cos \eta + d_1t^2+d_2=0 \label{E:eq1}\\
 t &\cos(\eta+\omega)+d_3t^2+d_4=0, \label{E:eq2}
 \end{align}
 where
 \begin{align}
&\omega \triangleq 2(\phi_{22}+\phi_{11}-\phi_{12}-\phi_{21}), \ \eta \triangleq \theta+2(\phi_{12}-\phi_{11}), \label{E:thetaEta}\\
&d_1\triangleq \frac{a_1|h_{12}|^4-1}{2a_1|h_{11}|^2|h_{12}|^2C_{x_1}^{\star}}, \
d_2\triangleq \frac{a_1|h_{11}|^4C_{x_1}^{\star 2}+b_1}{2a_1|h_{11}|^2|h_{12}|^2C_{x_1}^{\star}}, \notag \\
&d_3\triangleq \frac{|h_{22}|^2}{2|h_{21}|^2C_{x_1}^{\star}},\qquad
d_4=\frac{(a_2|h_{21}|^4-1)C_{x_1}^{\star 2}+b_2}{2a_2|h_{21}|^2|h_{22}|^2C_{x_1}^{\star}}.\notag
\end{align}
From \eqref{E:eq2}, we have
\begin{align}
& t\sin\eta\sin\omega=t\cos\eta\cos\omega+d_3t^2+d_4 \Rightarrow \notag \\
&  t^2(1-\cos^2\eta)\sin^2\omega=(t\cos\eta\cos\omega+d_3t^2+d_4)^2 \label{E:eq3}
\end{align}
Solving $\cos \eta$ from \eqref{E:eq1}, we have
\begin{align}
\cos \eta=-(d_1t^2+d_2)/t. \label{E:eq4}
\end{align}
Substituting \eqref{E:eq4} into \eqref{E:eq3} gives the following fourth order polynomial equation with respect to $t$:
\begin{align}
 [t^2 -(d_1t^2+d_2)^2]\sin^2\omega=[(d_3-d_1\cos\omega)t^2+d_4-d_2\cos\omega]^2.\notag
\end{align}
Since the above equation only has terms involving $t^2$, it can be transformed to the following quadratic equation by setting $z=t^2$, i.e.,
\begin{align}
&e_1z^2+e_2z+e_3=0,\\
\text{where }
e_1=&d_3^2+d_1^2-2d_1d_3\cos\omega, \notag \\
e_2=&2(d_1d_2+d_3d_4)-2(d_1d_4+d_2d_3)\cos\omega-\sin^2\omega,\notag \\
e_3=&d_2^2+d_4^2-2d_2d_4\cos\omega.\notag
\end{align}
Then $z$ can be easily solved. Since $z=t^2$ and $t$ is the magnitude of $\xC_{x_2}$, only the solutions of $z$ that are real and satisfy $0 \leq z \leq C_{x_2}^{\star 2}$ need to be kept, whereby the values for $t$ are obtained. For those values of $t$ satisfying $|(d_1t^2+d_2)/t|\leq 1$, we can get the value for $\eta$ based on \eqref{E:eq4}, i.e., $\eta=\arccos [-(d_1t^2+d_2)/t]$ or $\eta=2\pi-\arccos [-(d_1t^2+d_2)/t]$. Then $\theta$ can be obtained from \eqref{E:thetaEta}. If no such solutions exist, then $\Theta_{\mathcal{A}}$ is set to empty. 

 \bibliographystyle{IEEEtran}
\bibliography{IEEEabrv,IEEEfull}

\begin{thebibliography}{10}
\providecommand{\url}[1]{#1}
\csname url@samestyle\endcsname
\providecommand{\newblock}{\relax}
\providecommand{\bibinfo}[2]{#2}
\providecommand{\BIBentrySTDinterwordspacing}{\spaceskip=0pt\relax}
\providecommand{\BIBentryALTinterwordstretchfactor}{4}
\providecommand{\BIBentryALTinterwordspacing}{\spaceskip=\fontdimen2\font plus
\BIBentryALTinterwordstretchfactor\fontdimen3\font minus
  \fontdimen4\font\relax}
\providecommand{\BIBforeignlanguage}[2]{{%
\expandafter\ifx\csname l@#1\endcsname\relax
\typeout{** WARNING: IEEEtran.bst: No hyphenation pattern has been}%
\typeout{** loaded for the language `#1'. Using the pattern for}%
\typeout{** the default language instead.}%
\else
\language=\csname l@#1\endcsname
\fi
#2}}
\providecommand{\BIBdecl}{\relax}
\BIBdecl

\bibitem{272}
A.~B. Carleial, ``A case where interference does not reduce capacity,''
  \emph{{IEEE} Trans. Inf. Theory}, vol.~21, no.~5, pp. 569--570, Sep. 1975.

\bibitem{164}
T.~S. Han and K.~Kobayashi, ``A new achievable rate region for the interference
  channel,'' \emph{{IEEE} Trans. Inf. Theory}, vol.~27, no.~1, pp. 49--60, Jan.
  1981.

\bibitem{163}
R.~Etkin, D.~Tse, and H.~Wang, ``Gaussian interference channel capacity to
  within one bit,'' \emph{{IEEE} Trans. Inf. Theory}, vol.~54, no.~1, pp.
  5534--5562, Dec. 2008.

\bibitem{278}
X.~Shang, G.~Kramer, and B.~Chen, ``A new outer bound and noisy-interference
  sum-rate capacity for the {G}aussian interference channels,'' \emph{{IEEE}
  Trans. Inf. Theory}, vol.~55, no.~2, pp. 689--699, Feb. 2009.

\bibitem{344}
M.~Chiang, P.~Hande, T.~Lan, and C.~W. Tan, \emph{Power Control in Wireless
  Cellular Networks}.\hskip 1em plus 0.5em minus 0.4em\relax Foundations and
  Trends in Networking, 2008.

\bibitem{253}
A.~Gjendemsj$\o$, D.~Gesbert, G.~E. $\O$ien, and S.~G. Kiani, ``Binary power
  control for sum rate maximization over multiple interfering links,''
  \emph{{IEEE} Trans. Wireless Commun.}, vol.~7, no.~8, pp. 3164--3173, Aug.
  2008.

\bibitem{301}
Z.-Q. Luo and S.~Zhang, ``Dynamic spectrum management: complexity and
  duality,'' \emph{{IEEE} J. Sel. Topics Signal Process.}, vol.~2, no.~1, pp.
  57--73, Feb. 2008.

\bibitem{302}
F.~R. Farrokhi, K.~J.~R. Liu, and L.~Tassiulas, ``Transmit beamforming and
  power control for cellular wireless systems,'' \emph{{IEEE} J. Sel. Areas
  Commun.}, vol.~16, no.~8, pp. 1437--1450, Oct. 1998.

\bibitem{305}
E.~Visotsky and U.~Madhow, ``Optimal beamforming using transmit antenna
  arrays,,'' in \emph{Proc. IEEE Veh. Technol. Conf.}, vol.~1, Houston, Texas,
  May 1999, pp. 851--856.

\bibitem{306}
M.~Schubert and H.~Boche, ``Solution of the multiuser downlink beam-forming
  problem with individual {SINR} constraints,'' \emph{{IEEE} Trans. Veh.
  Technol.}, vol.~53, no.~1, pp. 18--28, Jan. 2004.

\bibitem{307}
H.~Dahrouj and W.~Yu, ``Coordinated beamforming for the multicell multi-antenna
  wireless system,'' \emph{{IEEE} Trans. Wireless Commun.}, vol.~9, no.~5, pp.
  1748--1759, May 2010.

\bibitem{308}
B.~Song, R.~Cruz, and B.~Rao, ``Network duality for multiuser {MIMO}
  beamforming networks and applications,'' \emph{{IEEE} Trans. Commun.},
  vol.~55, no.~3, pp. 618--630, Mar. 2007.

\bibitem{214}
A.~Wiesel, Y.~C. Eldar, and S.~Shamai~(Shitz), ``Linear precoding via conic
  optimization for fixed {MIMO} receivers,'' \emph{{IEEE} Trans. Signal
  Process.}, vol.~54, no.~1, pp. 161--176, Jan. 2006.

\bibitem{309}
M.~Bengtsson and B.~Ottersten, ``Optimal downlink beamforming using
  semidefinite optimization,'' in \emph{Proc. 37th Allerton Conf. on
  Commun.,Control, and Computing}, Mar. 1999, pp. 987--996.

\bibitem{325}
Y.~F. Liu, Y.~H. Dai, and Z.-Q. Luo, ``Coordinated beamforming for {MISO}
  interference channel: complexity analysis and efficient algorithms,''
  \emph{{IEEE} Trans. Signal Process.}, vol.~59, no.~3, pp. 1142--1156, Mar.
  2011.

\bibitem{313}
J.~Huang, R.~A. Berry, and M.~L. Honig, ``Distributed interference compensation
  for wireless networks,'' \emph{{IEEE} J. Sel. Areas Commun.}, vol.~24, no.~5,
  pp. 1074--1084, May 2006.

\bibitem{329}
R.~Zakhour and D.~Gesbert, ``Coordination on the {MISO} interference channel
  using the virtual {SINR} framework,'' in \emph{Proc. ITG/IEEE Work-shop Smart
  Antennas}, Feb. 2009.

\bibitem{233}
S.~Ye and R.~S. Blum, ``Optimized signaling for {MIMO} interference systems
  with feedback,'' \emph{{IEEE} Trans. Signal Process.}, vol.~51, no.~11, pp.
  2839--2848, Nov. 2003.

\bibitem{158}
H.~Sung, K.~J. Lee, S.~H. Park, and I.~Lee, ``Linear precoder designs for
  {$K$}-user interference channels,'' \emph{{IEEE} Trans. Wireless Commun.},
  pp. 291--301, Jan. 2010.

\bibitem{321}
L.~P. Qian, Y.~J. Zhang, and J.~Huang, ``{MAPEL}: achieving global optimality
  for a non-convex wireless power control problem,'' \emph{{IEEE} Trans.
  Wireless Commun.}, vol.~8, no.~3, pp. 1553--1563, Mar. 2009.

\bibitem{322}
E.~A. Jorswieck and E.~G. Larsson, ``Monotonic optimization framework for the
  two-user miso interference channel,'' \emph{{IEEE} Trans. Commun.}, vol.~58,
  no.~7, pp. 2159--2168, Jul. 2010.

\bibitem{320}
L.~Liu, R.~Zhang, and K.~C. Chua, ``Achieving global optimality for weighted
  sum-rate maximization in the {$K$}-user {G}aussian interference channel with
  multiple antennas,'' \emph{{IEEE} Trans. Wireless Commun.}, vol.~11, no.~5,
  pp. 1933--1945, May 2012.

\bibitem{345}
W.~Utschick and J.~Brehmer, ``Monotonic optimization framework for coordinated
  beamforming in multicell networks,'' \emph{{IEEE} Trans. Signal Process.},
  vol.~60, no.~4, pp. 1899--1909, Apr. 2012.

\bibitem{326}
E.~Bj\"ornson, G.~Zheng, M.~Bengtsson, and B.~Ottersten, ``Robust monotonic
  optimization framework for multicell {MISO} systems,'' \emph{{IEEE} Trans.
  Signal Process.}, vol.~60, no.~5, pp. 2508--2523, May 2012.

\bibitem{333}
S.~S. Christensen, R.~Agarwal, E.~D. Carvalho, and J.~M. Cioffi, ``Weighted
  sum-rate maximization using weighted {MMSE} for {MIMO-BC} beamforming
  design,'' \emph{{IEEE} Trans. Wireless Commun.}, vol.~7, no.~12, pp.
  4792--4799, Dec. 2008.

\bibitem{237}
M.~Razaviyayn, M.~Sanjabi, and Z.-Q. Luo, ``Linear transceiver design for
  interference alignment: complexity and computation,'' \emph{{IEEE} Trans.
  Inf. Theory}, vol.~58, no.~5, pp. 2896--2910, May 2012.

\bibitem{335}
Q.~Shi, M.~Razaviyayn, Z.-Q. Luo, and C.~He, ``An iteratively weighted {MMSE}
  approach to distributed sum-utility maximization for a {MIMO} interfering
  broadcast channel,'' \emph{{IEEE} Trans. Signal Process.}, vol.~59, no.~9,
  pp. 4331--4340, Sep. 2011.

\bibitem{316}
G.~Scutari, P.~Palomar, and S.~Barbarossa, ``Competitive design of multiuser
  {MIMO} systems based on game theory: a unified view,'' \emph{{IEEE} J. Sel.
  Areas Commun.}, vol.~26, no.~7, pp. 1089--1103, Sep. 2008.

\bibitem{92}
A.~Host-Madsen and A.~Nosratinia, ``The multiplexing gain of wireless
  networks,'' \emph{Int. Symp. on Inf. Theory and Its Applications}, pp.
  2065--2069, 4-9 Sept. 2005.

\bibitem{85}
V.~R. Cadambe and S.~A. Jafar, ``Interference alignment and degrees of freedom
  of the {$K$}-user interference channel,'' \emph{{IEEE} Trans. Inf. Theory},
  vol.~54, no.~8, pp. 3425--3441, Aug. 2008.

\bibitem{177}
S.~W. Peters and R.~W. Heath~Jr., ``Cooperative algorithms for {MIMO}
  interference channels,'' \emph{{IEEE} Trans. Veh. Technol.}, vol.~60, no.~1,
  pp. 206 -- 218, Jan. 2011.

\bibitem{187}
D.~A. Schmidt, C.~Shi, R.~A. Berry, M.~L. Honig, and W.~Utschick, ``Minimum
  mean squared error interference alignment,'' \emph{IEEE Asilomar Conference
  on Signals, Systems and Computers ({ACSSC})}, pp. 1106 -- 1110, Nov. 2009.

\bibitem{252}
K.~Gomadam, V.~R. Cadambe, and S.~A. Jafar, ``A distributed numerical approach
  to interference alignment and applications to wireless interference
  networks,'' \emph{{IEEE} Trans. Inf. Theory}, vol.~57, no.~6, pp. 3309--3322,
  May 2011.

\bibitem{250}
M.~Charafeddine, A.~Sezgin, and A.~Paulraj, ``Rates region frontiers for
  $n$-user interference channel with interference as noise,'' in \emph{Proc.
  {A}llerton {C}onference}, Sep. 2007.

\bibitem{249}
E.~Jorswieck, E.~Larsson, and D.~Danev, ``Complete characterization of the
  {P}areto boundary for the {MISO} interference channel,'' \emph{{IEEE} Trans.
  Signal Process.}, vol.~56, no.~10, pp. 5292--5296, Oct. 2008.

\bibitem{240}
R.~Zhang and S.~Cui, ``Cooperative interference management with {MISO}
  beamforming,'' \emph{{IEEE} Trans. Signal Process.}, vol.~58, no.~10, pp.
  5450--5458, Oct. 2010.

\bibitem{243}
X.~Shang, B.~Chen, and H.~V. Poor, ``Multiuser {MISO} interference channles
  with single-user detection: optimality of beamforming and the achievable rate
  region,'' \emph{{IEEE} Trans. Inf. Theory}, vol.~57, no.~7, pp. 4255 -- 4273,
  Jul. 2011.

\bibitem{323}
R.~Mochaourab and E.~A. Jorswieck, ``Optimal beamforming in interference
  networks with perfect local channel information,'' \emph{{IEEE} Trans. Signal
  Process.}, vol.~59, no.~3, pp. 1128--1141, Mar. 2011.

\bibitem{245}
F.~D. Neeser and J.~L. Massey, ``Proper complex random processes with
  applications to information theory,'' \emph{{IEEE} Trans. Inf. Theory},
  vol.~39, no.~4, pp. 1293 -- 1302, Jul. 1993.

\bibitem{246}
P.~J. Schreier and L.~L. Scharf, \emph{Statistical {S}ignal {P}rocessing of
  {C}omplex-{V}alued {D}ata: {T}he {T}heory of {I}mproper and {N}oncircular
  {S}ignals}.\hskip 1em plus 0.5em minus 0.4em\relax Cambridge (UK): Cambridge
  Univ. Press, 2010.

\bibitem{248}
G.~Taub\"{o}ck, ``Complex-valued random vectors and channels: entropy,
  divergence, and capacity,'' \emph{{IEEE} Trans. Inf. Theory}, vol.~58, no.~5,
  pp. 2729--2744, May 2012.

\bibitem{348}
S.~A. Jafar, \emph{Interference Alignment: A New Look at Signal Dimensions in a
  Communication Network}.\hskip 1em plus 0.5em minus 0.4em\relax Foundations
  and Trends in Communications and Information Theory, 2010, vol.~7, no.~1.

\bibitem{189}
V.~R. Cadambe, S.~A. Jafar, and C.~Wang, ``Interference alignment with
  asymmetric complex signaling - settling the {H}ost-{M}adsen-{N}osratinia
  conjecture,'' \emph{{IEEE} Trans. Inf. Theory}, pp. 4552 -- 4565, Sep. 2010.

\bibitem{255}
Z.~K.~M. Ho and E.~Jorswieck, ``Improper {G}aussian signaling on the two-user
  {SISO} interference channel,'' \emph{{IEEE} Trans. Wireless Commun.},
  vol.~11, no.~9, pp. 3194 -- 3203, Sep. 2012.

\bibitem{256}
S.~H. Park, H.~Park, and I.~Lee, ``Coordinated {SINR} balancing techniques for
  multi-cell downlink transmission,'' in \emph{Proc. VTC 2010-fall}, 2010.

\bibitem{262}
P.~J. Schreier and L.~L. Scharf, ``Second-order analysis of improper complex
  random vectors and processes,'' \emph{{IEEE} Trans. Signal Process.},
  vol.~51, no.~3, pp. 714--725, Mar. 2003.

\bibitem{349}
Z.-Q. Luo, W.~K. Ma, A.~M. So, Y.~Ye, and S.~Zhang, ``Semidefinite relaxation
  of quadratic optimization problems,'' \emph{IEEE Signal Processing Mag},
  vol.~27, no.~3, pp. 20--34, May 2010.

\bibitem{361}
R.~A. Horn and C.~R. Johnson, \emph{Matrix Analysis}, 2nd~ed.\hskip 1em plus
  0.5em minus 0.4em\relax {New York}: {C}ambridge {U}niv. {P}ress, 2013.

\bibitem{269}
B.~Picinbono and P.~Chevalier, ``Widely linear estimation with complex data,''
  \emph{{IEEE} Trans. Signal Process.}, vol.~43, no.~8, pp. 2030--2033, Aug.
  1995.

\bibitem{202}
S.~Boyd and L.~Vandenberghe, \emph{Convex {O}ptimization}.\hskip 1em plus 0.5em
  minus 0.4em\relax Cambridge, {U.K.}: {C}ambridge {U}niv. {P}ress, 2004.

\bibitem{93}
S.~A. Jafar and M.~Fakhereddin, ``Degrees of freedom for the {MIMO}
  interference channel,'' \emph{{IEEE} Trans. Inf. Theory}, vol.~53, no.~7, pp.
  2637--2642, Jul. 2007.

\bibitem{251}
A.~Hj{\o}rungnes, \emph{Complex-{V}alued {M}atrix {D}erivatives {W}ith
  {A}pplications in {S}ignal {P}rocessing and {C}ommunications}.\hskip 1em plus
  0.5em minus 0.4em\relax Cambridge (UK): Cambridge Univ. Press, 2011.

\end{thebibliography}
\end{document}